\documentclass[11pt]{article}
\usepackage[T1]{fontenc}

\usepackage{graphicx}
\usepackage{amsthm}
\usepackage{tabularx}
\usepackage{amssymb}
\usepackage{ccaption}
\usepackage[toc,page]{appendix}
\usepackage{subcaption}
\usepackage{float}
\usepackage{bm}
\usepackage{mathtools}
\usepackage{hyperref}
\usepackage[square,sort&compress,numbers]{natbib}
\usepackage{bbold}
\usepackage{fancyhdr}

\hypersetup{
    colorlinks,
    citecolor=black,
    filecolor=black,
    linkcolor=black,
    urlcolor=black
}

\newcommand*{\myeqref}[2][]{%
  \hyperref[{#2}]{#1(\ref*{#2})}%
}
\def\equationautorefname#1#2\null{%
  #1(#2\null)%
}

\newcommand{\R}{\mathbb{R}}
\newcommand{\N}{\mathbb{N}}
\newcommand{\cG}{\mathcal{G}}

\newcommand{\lp}{\left( }
\newcommand{\rp }{\right) }

\numberwithin{equation}{section}

\usepackage[left=2.5cm,right=2.5cm,top=2.5cm,bottom=2.5cm]{geometry}
\usepackage{array}

\newcolumntype{C}{>{\raggedright\arraybackslash}p{1.5cm}<{}}

\newtheorem{theorem}{Theorem}[section]
\newtheorem{corollary}{Corollary}[section]

\newtheorem{assumption}{Assumption}[section]

\newtheorem{defn}{Definition}[section]

\title{Algebraic framework for determining laminar pattern bifurcations by lateral-inhibition in 2D and 3D bilayer geometries}
\author{Joshua W. Moore\footnote{Correspondence to Joshua W. Moore: moorej16@cardiff.ac.uk}$\, \, ^{,1}$, Trevor C. Dale$^{2}$ \& Thomas E. Woolley$^{1}$}    
\date{%
    \small $^1$School of Mathematics, Cardiff University, Senghennydd Rd, Cardiff, CF24 4AG, UK\\%
   \small $^2$School of Biosciences, Cardiff University, Museum Ave, Cardiff, CF10 3AX, UK\\%
}

\begin{document}
\maketitle

\vspace{-0cm}
\begin{center}
\large{\textbf{Abstract}}
\end{center}

Fine-grain patterns produced by juxtacrine signalling, have been studied using static monolayers as cellular domains. Unfortunately, analytical results are restricted to a few cells due to the algebraic complexity of nonlinear dynamical systems. Motivated by concentric patterning of Notch expression observed in the mammary gland, we combine concepts from graph and control theory to represent cellular connectivity. The resulting theoretical framework allows us to exploit the symmetry of multicellular bilayer structures in 2D and 3D, thereby deriving analytical conditions that drive the dynamical system to form laminar patterns consistent with the formulation of cell polarity. Critically, the conditions are independent of the precise dynamical details, thus the framework allows for the utmost generality in understanding the influence of cellular geometry on patterning in lateral-inhibition systems. Applying the analytic conditions to mammary organoids suggests that intense cell signalling polarity is required for the maintenance of stratified cell-types within a static bilayer using a lateral-inhibition mechanism. Furthermore, by employing 2D and 3D cell-based models, we highlight that the cellular polarity conditions derived from static domains have the capacity to generate laminar patterning in dynamic environments. However, they are insufficient for the maintenance of patterning when subjected to substantial morphological perturbations. In agreement with the mathematical implications of strict signalling polarity induced on the cells, we propose an adhesion dependent Notch-Delta biological process which has the potential to initiate bilayer stratification in a developing mammary organoid.

\section{Introduction}

Lateral feedback is considered a fundamental driving process for the emergence of fine-grain pattern formation \cite{GilbertScottF.2016Db}. Such patterning is critical in the development of many multicellular biological systems such as \textit{Drosophila} eye formation, murine hair organisation in auditory epithelia and establishing blood vessels during human embryogenesis \cite{cagan2009principles,togashi2011nectins,kim2015notch}. In contrast to the approach of using reaction-diffusion systems to generate spatially continuous patterns \cite{turing1952chemical}, systems of ordinary differential equations (ODEs) can be used to generate a discretised description of the space, enabling the formation of fine-grain patterns on the resolution of individual cells. These discrete spatial ODE models seek to emulate the behaviour of contact-dependent cell-cell signalling mechanism known as juxtacrine signalling, a common form of cellular communication in epithelial tissue \cite{GilbertScottF.2016Db}. \par

The juxtacrine signalling mechanism relies on membrane bound signal proteins on a sender cell binding to surface anchored receptors on a receiving cell, imposing a contact-dependence \cite{GilbertScottF.2016Db}. Critically, cells can only use juxtacrine signalling to communicate with their direct neighbours in the absence of receptor extensions \cite{Hadjivasiliou2016}, as demonstrated in \autoref{fig:Notch_diagram}a. Consequently, the spatial organisation of cells is of fundamental importance in orchestrating signal protein patterning required for specific organ development \cite{manz2010spatial}. \par

\begin{figure}[h!]
         \centering
         \includegraphics[width=0.6\textwidth]{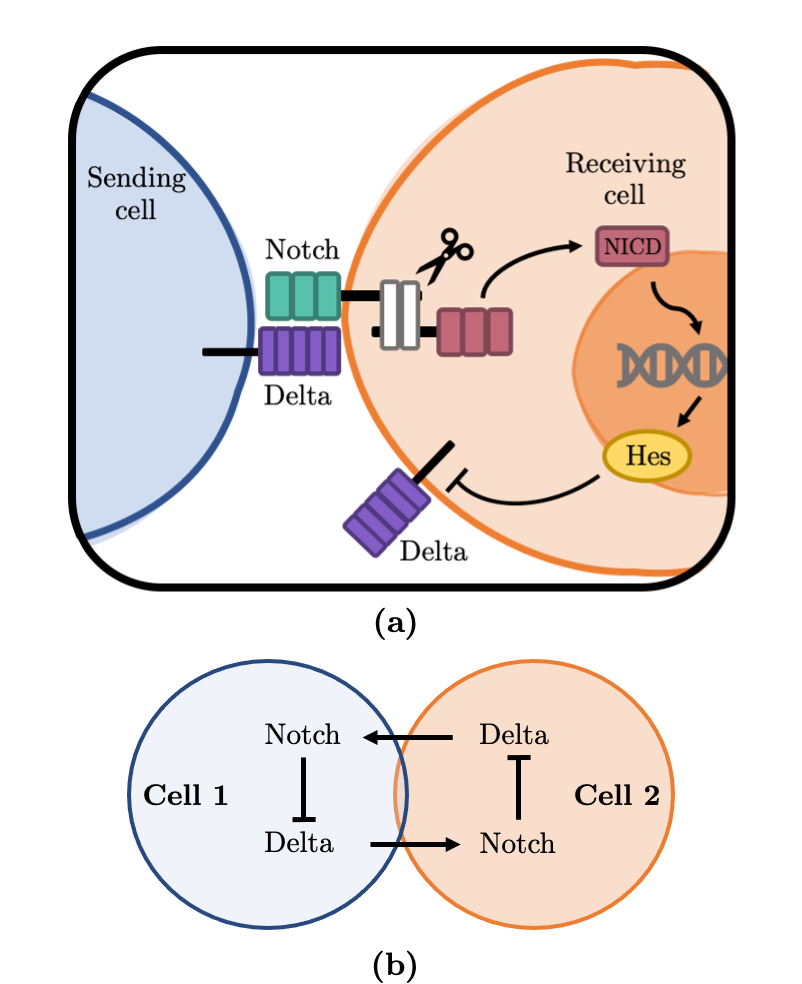}

 			\caption{A schematic diagram of the canonical Notch pathway as an example juxtacrine signalling mechanism. (a) Membrane bound Delta ligands (purple rectangles) on a signal sending cell bind to membrane bound Notch receptors (green rectangles) on a receiving cell. The activation of Notch receptors initiates the cleavage of Notch into the cytosol of the receiving cell, known as NICD. The NICD then translocates to the nucleus where it promotes the transcription of Hes, an inhibitor of Delta ligand targets. Adapted from \cite{Bray2006}. (b) A minimal representation of the negative feedback dynamics of Notch and Delta in coupled cells. This mathematical simplification was first conceived by Collier et al. (1996) \cite{Collier1996}.}
        \label{fig:Notch_diagram}
\end{figure}

Mathematically, juxtacrine pattern formation has been extensively studied over the last two decades \cite{Wearing2001,Webb2004,Webb2004_osc,Hadjivasiliou2016,Formosa-Jordan2014}, commonly focusing on lateral-inhibition mechanisms as the underlying biological process. An overarching conclusion from the family of papers focused on juxtacrine pattern analysis of lateral-inhibition models is that linear analysis techniques are insufficient to determine precise conditions for patterning, and are only able to predict the existence of patterning \cite{Wearing2001}. In light of this, there has been a reliance on numerical simulations to elucidate the parameter regimes required to produce patterning of various wavelengths. However, the lateral-inhibition model parameters are not the only factors influencing the emergence of patterns. The geometry of the cellular domain (epithelial sheet) on which the juxtacrine model is being applied has a large impact on the obtainable patterning. This was highlighted by Webb et al. (2004), where they compared a honeycomb domain to a simple grid domain in 2D under a standard four point connectivity stencil for cellular connectivity (see \autoref{fig:checkerboard_domains} and \autoref{fig:grid_neighbourhoods}a). In doing so, they show the considerable differences in parameter regimes required to achieve similar pattern formations for each domain type \cite{Webb2004_osc}. Moreover, cellular asymmetries have the capacity to produce unique patterning, which are unobtainable on regular domains, but are more biologically realistic \cite{Webb2004}. In support of this, when a lateral-inhibition model was coupled with a mechanism for cellular protrusions to encourage contact over larger distances, a large family of distinct patterns were observed over a regular 2D honeycomb grid \cite{Hadjivasiliou2016}. Although, to achieve the striped patterning that is observed in various biological systems such as in the mammary gland and zebrafish skin pigments \cite{Lilja2018,hamada2014involvement}, cellular protrusions must be directed perpendicular to the stripe, indicating a requirement for cellular polarity.\par

An alternative approach to pattern formation analysis in lateral-inhibition models was introduced by Arcak \cite{Arcak2013}, where they viewed cells as vertices on a connected graph that interact using dynamical input-output systems. This abstract approach produced analytical conditions for the existence and stability of checkerboard patterning in cyclic domains with an unbound number of cells, as demonstrated in \autoref{fig:checkerboard_domains}. Therefore, extending the analysis conducted in \cite{Wearing2001,Webb2004,Webb2004_osc} that was restricted to only two cells due to the complexity of the systems studied. Moreover, the graph theoretic approach to juxtacrine systems was further refined when graph partitioning was applied to represent patterning within collections of cells \cite{RufinoFerreira2013}, generalising the previous results of \cite{Arcak2013}, which developed a framework to prove the existence and stability of a family of patterns within periodic domains in both grid and hexagonal lattices. Such previous work emphasises the relationship between how cells are connected and the obtainable patterns. Nonetheless, these conditions were derived using static domains and were heavily dependent on a number of assumptions regarding the graph's topology (reviewed in \autoref{sec:graph_theory}). These assumptions cannot always be adhered to when investigating patterning on an evolving biological system, although they may be true in certain quasi-steady stages of its development.\par

 \begin{figure}[h!]
         \centering
         \includegraphics[width=0.6\textwidth]{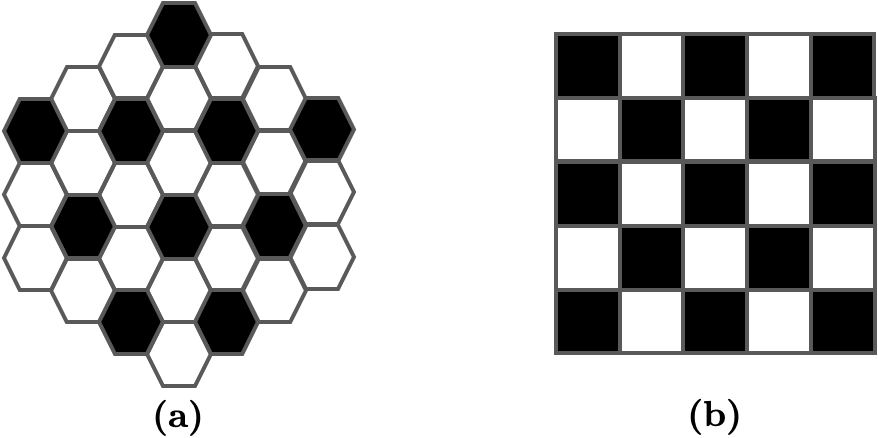}

 			\caption{Diagrams of checkerboard patterns of protein expression via lateral feedback in (a) honeycomb and (b) grid cellular domains \cite{Collier1996}.}
        \label{fig:checkerboard_domains}
\end{figure}

The canonical Notch pathway is a well-studied example of a juxtacrine signalling pathway with an essential role in cell fate determination and morphogenic bifurcations in developmental systems \cite{Lilja2018,Jorgensen2018,Fre2005,Lafkas2015}. The Notch pathway describes a lateral-inhibition mechanism as the release of Notch-intracellular-domain (NICD), via Notch receptor activation, has a downstream negative affect on Delta production. In turn, the Delta production inhibits Notch activation in neighbouring cells and encourage a negative feedback loop, see \autoref{fig:Notch_diagram}a. The release of Notch from the membrane into the cytosol triggers the transcription of members of the ``Hairy-Enhancer of Split'' (Hes) superfamily, in which target genes repress lineage specific determinants. Mammals exhibit four paralogues of the Notch receptor, Notch1 to Notch4, with associated Delta ligands that each observe the autoregulation mechanism outlined by the canonical pathway \cite{Lloyd-Lewis2019}. A more detailed description of the canonical Notch pathway can be found elsewhere \cite{Bray2006}. \par

A particular biological system that is highly dependent on the Notch pathway is the mammary organoid. Mammary organoids are three-dimensional tissue cultures that are currently the most accurate representation of \textit{in vivo} mammary gland biology \cite{Simian2017}. Throughout its development, the mammary organoid retains a consistent bilayer structure of cells as seen in \autoref{fig:mammary_organoid_intro}. That is, the outer layer holds the elongated, contractile basal cells, whereas the inner layer consists of cuboidal luminal cells. Once these layers have been established, a hollow lumen forms, surrounded by the bilayer of cells. \par
Notch1 signalling (denoted by Notch signalling hereafter) is a critical determinant of luminal cell differentiation in mammary epithelial cells (MECs) \cite{Lilja2018}. It has been established that Notch activation is required to support differentiation of the basal stem cells to the luminal population in the mammary organoid and therefore it is a key component in the maintenance of a developing mammary system \cite{Bouras2008}. In addition, sudden Notch activation within the luminal cells coincides with the locations of symmetry breaking events of embryonic MECs \cite{Lilja2018}. Thus, it has been hypothesised that Notch activation via basal cells, or contact with the basement membrane, is required to develop branched epithelia \cite{Lloyd-Lewis2019}. \par

During any stage of development of the mammary gland and organoid, MECs are capable of self-organising to form an outer layer of cells that highly express Delta (low Notch), and in contrast, inner layers of cells that surround a hollow lumen that express low Delta (high Notch), see Figures 3b-3d \cite{Bouras2008, Lafkas2013}. It is unclear whether this spatial patterning is a consequence or cause of the morphology of developing mammary ducts, although, it is clear that the concentric (laminar) patterning of the bilayer of cells is robust to morphological perturbations.\par

\begin{figure}[h!]
         \centering
         \includegraphics[width=1\textwidth]{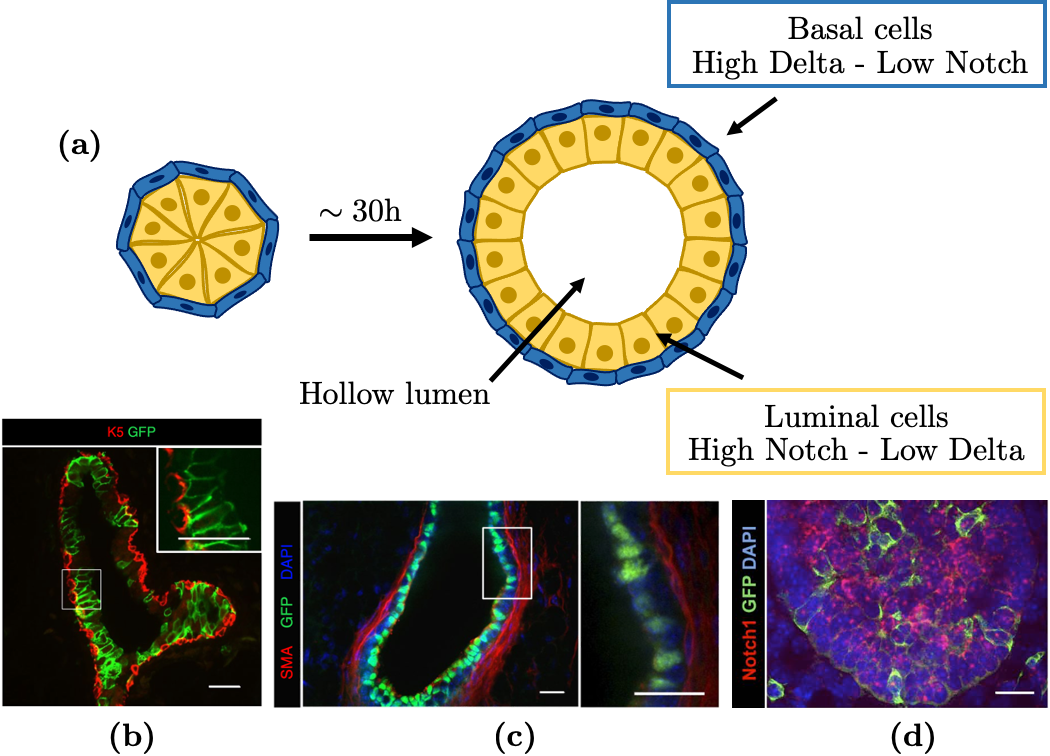}

 			\caption{The structure of a mammary organoid and the spatial distribution of Notch1 expression. (a) A simple 2D diagram of the structure of a developing mammary organoid, highlighting lumen formation and maintenance of stratified bilayer. (b-c) Cross-sections of mammary ducts of 6 week old mice, where Notch1-derived lineages are labelled in green by (b) membrane bound green fluorescent protein (GFP) and (c) nuclear GFP. (d) Representative sections of embryonic mammary buds. The small red dots highlight the presence of Notch1 protein which are clustered towards the centre of the bud.  Scale bar, $20\mu$m ($10\mu$m in magnifications).  Images reproduced from Lilja et al. 2018 \cite{Lilja2018}.}
        \label{fig:mammary_organoid_intro}
\end{figure}

Conditions defining a bilayer laminar pattern have yet to be derived using a simple mathematical lateral-inhibition model in spheroid geometries. Here, we apply our general framework to a previously developed ODE model of Notch-Delta, and obtain conditions on Delta transmission that are sufficient for the bilayer laminar patterns to form, which are in agreement with experimental observations. \par

We first construct our system of cells using the connected graph approach, which allows us to derive analytical conditions on cell-type dependent weightings of Delta cell-cell transmission and thus, predict the existence and stability of bilayer patterning by exploiting existing results in graph and control theory \cite{Arcak2013, RufinoFerreira2013}. Then by numerically comparing a variety of fixed cellular geometries in 2D and 3D, we provide a further restriction on the analytical bound to produce the desired patterning, thereby deriving an empirical relationship between signal transmission weightings and the cellular connectivity for patterning, independent of spatial dimension. In particular, the condition applies to fixed geometries in both 2D and 3D. We then show by integration of the Notch-Delta model (NDM) into a lattice free cell-based model using the cell-centre framework, how the transient transitions of cellular connectivity in a dynamic domain can alter patterning.

\newpage

\section{Theory}

We develop a framework to investigate cellular signalling polarity that is independent of a lateral-inhibition ODE system. Though, to elucidate the dependence of cell-type transmission of Delta in a bilayer structures we consider the original lateral-inhibition ODE model constructed by Collier et al. (1996) as an example system \cite{Collier1996}. By adapting the spatial averaging term to include cell-type dependent weightings on Delta transmission, we impose heterogeneity within the cellular system to promote the emergence and stability of bilayer laminar pattern formation of Notch-Delta that is observed experimentally in Figures 3b-3d. In addition, we introduce the notion of the graphical representation of cellular connectivity and a framework for cellular coupling to bridge the geometry of the fixed lattice to the dynamics of the ODE model.

\subsection{Lateral-inhibition model}\label{sec:lateral_inhibition_model}
The NDM developed by Collier et al. (1996) was the first explicit lateral-inhibition model that was used to investigate fine-grain patterns that are observed in a variety of biological systems \cite{Collier1996}. The intracellular model contains only two components, Notch ($N$) and Delta ($D$) activation, simplifying the underlying biochemical processes, which allows freedom of interpretation of $N$  and  $D$. When studying the dynamics of Notch-Delta in the mammary organoid, we will consider $N$ to be the NICD active protein concentration within the cytosol and $D$ to be the amount of active membrane bound Delta ligands on the surface of the cell, see \autoref{fig:Notch_diagram}a. The inverse relationship between intracellular Notch and Delta is the key feature of the spatially discrete ODE model, which is described by the negative feedback loop depicted in \autoref{fig:Notch_diagram}b.
The non-dimensional NDM defined by Collier et al. (1996) is based upon the following assumptions:
\begin{enumerate}
\item Cells interact through Delta–Notch signalling only with cells with which they are in direct contact, that is, adhering to the juxtacrine mechanism.
\item The rate of production of Notch activity is an increasing function of the level of Delta activity in neighbouring cells.
\item The rate of production of Delta activity is a decreasing function of the level of activated Notch in the same cell.
\item Production of Notch and Delta activity is balanced by decay, described by simple exponential decay with fixed rate constants.
\item The activity of Notch and Delta are uniformly distributed throughout the cell.
\item Instantaneous transcription of downstream Notch targets such that the model assumes no delay in Notch and Delta interactions.
\end{enumerate}
These assumptions outline the characteristics of a general lateral-inhibition model, which can be formalised mathematically as,

\begin{align}
\dot{N}_{i} &=\underbrace{ f\left( \langle D_{i} \rangle \right)}_{ \substack{\text{NICD activation via} \\ \text{Delta binding from}\\ \text{adjacent cells} } } - \underbrace{\mu_{1}N_{i}}_{ \substack{\text{NICD} \\ \text{degradation}} }, \\
\dot{D}_{i} &=\underbrace{ g\left( N_{i} \right) }_{ \substack{\text{Delta inhibition} \\ \text{by NICD}}}- \underbrace{\mu_{2}D_{i}}_{ \substack{\text{Delta} \\ \text{degradation}} },
\end{align}
where $f$ and $g$ are bounded increasing and decreasing functions respectively. These functions have the form,
\begin{equation}
f(x) = \frac{x^{k}}{a + x^{k}} \quad \text{and} \quad g(x) = \frac{1}{1+bx^{h}},
\end{equation}
where parameters $a,b,\mu_{1},\mu_{2}>0$ and hill coefficients $k,h\geq 1$. The subscript $i$ corresponds to cell identity within the system and the definition of the local spatial mechanism, $\langle D_{i} \rangle$, will be discussed in \autoref{Implementation section}.

\subsection{Graphical approach to cellular connectivity}
Rather than simply considering the discrete spatial system as an ODE network, we can represent the cellular connections as an undirected connected graph $\mathcal{G} = \mathcal{G}\left(V,E\right)$, where vertices $v \in V$ represent cells and edges $e \in E$ correspond to cellular connections, see \autoref{fig:example_weightings}. The vertices $v_{i}$ and $v_{j}$ representing cells $i$ and $j$ are considered to be connected if there exists an edge, $e_{i,j}$, between $v_{i}$ and $v_{j}$ such that $e_{i,j} \neq \emptyset$. Physically, we say that $e_{i,j} \neq \emptyset$ if the cell surfaces of cell $i$ and $j$ are touching. We represent the signal strength of cellular connectivity between cells $i$ and $j$ using nonnegative cell-type dependent weighting coefficients $w_{i,j}$. Namely, 
\begin{equation}
w_{i,j} = \left\{
	\begin{array}{ll}
		w_{1}  & \mbox{if } e_{i,j} \neq \emptyset \wedge \tau_{i} = \tau_{j}, \\
		w_{2} & \mbox{if }  e_{i,j} \neq \emptyset \wedge \tau_{i} \neq \tau_{j},  \\
		0 &  \mbox{if } e_{i,j} = \emptyset,
	\end{array}
\right.
\end{equation}
where $w_{1},w_{2} \in \R_{> 0}$ and cell-type of cell $i$ is denoted by $\tau_{i}$. Explicitly, let $w_{i,j} = 0$ if cells $i$ and $j$ are not connected, then if cells $i$ and $j$ are connected and of the same type, let $w_{i,j} = w_{1}$, and if cells $i$ and $j$ are connected and are different types we let $w_{i,j} = w_{2}$, see \autoref{fig:example_weightings}. Due to the symmetry of the undirected graph $\mathcal{G}$, we impose that $w_{i,j} = w_{j,i}$, which is physically realistic as cell-types are constant regardless of the perspective of the cell. The coefficients $w_{i,j}$ can used to mediate Delta transmission between adjacent cells dependent on cell-type and therefore will be the focus of our study.\par

We couple the notation of cellular connectivity with the lateral-inhibition model, as discussed in \autoref{sec:lateral_inhibition_model}, by constructing an \textit{adjacency matrix} for the system of $N$ cells (for $N \in 2 \N$ to account for bilayer geometries), $\bm{W} \in \R_{\geq 0}^{N\times N}$, using the cell-type weighting coefficients $w_{i,j}$ for each cell. This method of connectivity has been defined previously \cite{RufinoFerreira2013,Arcak2013}. Explicitly, we define the \textit{averaging} operator $\langle \cdot \rangle$ for the Delta transmission between adjacent cells in static geometries via
\begin{equation}\label{eqn:average_delta}
\langle \bm{D}(t) \rangle = \frac{1}{n_{1}w_{1} + n_{2}w_{2}}\bm{W}\bm{D}(t),
\end{equation}
where $\bm{D}(t)$ represents a vector of Delta concentrations for each cell in the system, $\bm{D}(t) = [D_{1}(t),...,D_{N}(t)]^{T}$, and $n_{1}$ corresponds to the number of cells of the same type adjacent to cell $i$, whereas $n_{2}$ is the number of cells adjacent of a different type, $n_{1},n_{2} \in \mathbb{N}$. The total cellular connectivity is defined as the total number of adjacent cells to each cell $i$, denoted by $N_{c} = n_{1} + n_{2}$. In addition, we introduce notation for the total scale weighting for each cell, $N_{w} = n_{1}w_{1} + n_{2}w_{2}$, for brevity. The inclusion of a scaling was introduced in \cite{Arcak2013} and enabled the direct comparison of cellular connectivity to a probability transition matrix of a reversible Markov chain, such that in each row $(1/N_{w})\sum_{j}w_{i,j} = 1$ for all $i$. \par

In this study, we make the assumption that there are only two cell types, basal cells and luminal cells, which are organised into separate layers. We consider cells to be located in either of the concentric layers, as seen in \autoref{fig:mammary_organoid_intro}a, and therefore we assume each static lattice is symmetric with periodic boundaries in 2D and 3D. Consequently, the matrix $\bm{W}$ can be decomposed into blocks of two smaller matrices $\bm{W}_{1},\bm{W}_{2} \in \R_{\geq  0}^{(N/2) \times (N/2)} $, such that $\bm{W}$ has the following form,
\begin{equation}\label{eqn:block_w}
\bm{W} = \left[
  \begin{array}{cc}
  \bm{W}_{1} & \bm{W}_{2} \\
   \bm{W}_{2} &  \bm{W}_{1} 
  \end{array}
\right].
\end{equation}
Where row $i$ of $\bm{W}_{1}$ represents the cellular connections of cell $i$ to adjacent cells of the same type, and the rows of  $\bm{W}_{2}$ corresponds to the cellular connections to cells of differing types. For example, for the standard orthoganal template, or Von Neumann neighbourhood, consisting of two layers of cells as shown in \autoref{fig:example_weightings}, has connectivity matrices,
\begin{equation}
\bm{W}_{1} =w_{1} \begin{bmatrix}
0 & 1 & 0 & 0 & \cdots & 0 & 1\\
1 & 0 & 1 & 0 & \cdots & 0 & 0 \\
0 & 1 & 0 & 1 & 0 & \cdots & 0 \\
\vdots & \ddots & \ddots & \ddots & \ddots & \ddots & \vdots \\
0 & \cdots & 0 & 1 & 0 & 1 & 0 \\
 0 & 0& \cdots & 0 & 1 & 0 & 1 \\
 1 & 0 & \cdots & 0 & 0 & 1 & 0 
\end{bmatrix} \quad \mbox{and} \quad \bm{W}_{2} =w_{2} \mathbb{1}_{N/2},
\end{equation}
where $\mathbb{1}_{N/2}$ is the identity matrix in $\left(N/2 \right)$-dimensions. The computational framework presented here enables the use of any symmetric and periodic static geometry of cellular systems that consist of representative cell types. 	We will present other possible connection geometries later (\autoref{fig:all connectivity diagrams}).

 \begin{figure}[H]
         \centering
         \hspace{1cm}
         \includegraphics[width=0.7\textwidth]{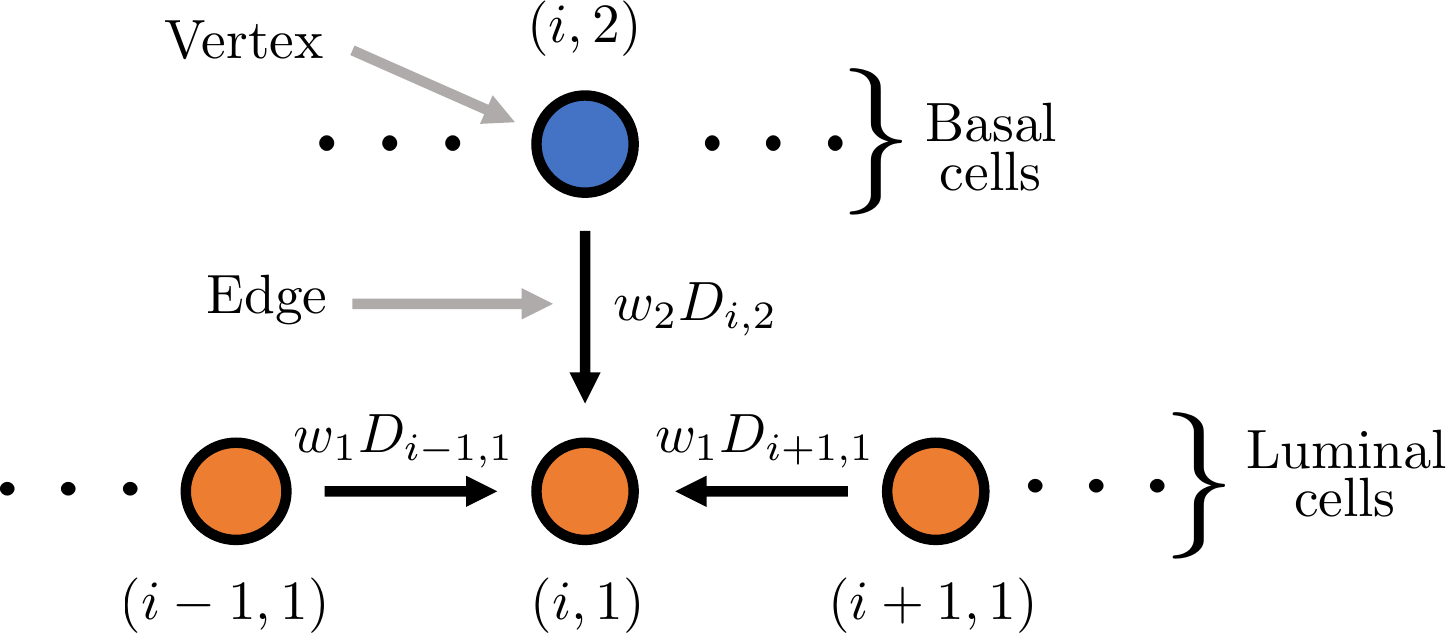}

 			\caption{An illustrative computational template for weighted averaging of Delta transmission from the perspective of luminal cell $i,1$ in a static lattice. In the present case we are considering a bilayer structure consisting of two cell-types. Thus, $j = 1,2$ and orange nodes identify luminal cells and blue nodes identify basal cells.}
        \label{fig:example_weightings}
\end{figure}

By representing cells as vertices in the connected graph $\mathcal{G}$, we are able to manipulate the geometry of the graph to investigate parameter regimes of $w_{1}$ and $w_{2}$, such that we obtain the desired patterning. Here, we investigate a variety of regular cyclic fixed lattices 2D and 3D. We assume that $e_{i,j} \neq \emptyset$ if cell $j$ lies within a circle (or sphere) of radius $\rho_{c}$ drawn around cell $i$, and the rest length of the lattice is unitary. The circle (or sphere) can be viewed as the cell membrane to which the Notch receptors and Delta ligands are anchored to. In addition, we introduce notation for the cell-type ratio for each cell, which is defined as,
\begin{equation}
R_{\tau} = \frac{\text{\# of adjacent cells of the same cell-type}}{\text{\# of adjacent cells of a different cell-type}} = \frac{n_{1}}{n_{2}},
\end{equation}
due to the symmetry of the domains, $R_{\tau}$ is homogeneous for all cells in the system. We chose three representative lattice structures in this study: (1) grid, (2) triangulated and (3) overlapped grid, to characterise the quasi-steady cellular configurations that may occur during the development mammary organoid. We then increase the connectivity radius, $\rho_{c}$, to obtain different neighbourhoods around each cell.\par
Focusing on grid domains, we examine two common cellular neighbourhoods used within the field of Cellular Automata. That is, taking $\rho_{c} = 1$ yields a Von Neumann neighbourhood, which is defined by a central node, surrounded by 4 other nodes in the north, east, south and west directions (\autoref{fig:grid_neighbourhoods}a) \cite{toffoli1987cellular}. Whereas increasing the connectivity spheres radius such that $\rho_{c} = \sqrt{2}$, we obtain a Moore neighbourhood, that includes the diagonal nodes missing from the Von Neumann neighbourhood (\autoref{fig:grid_neighbourhoods}b) \citep{toffoli1987cellular}. 

 \begin{figure}[h!]
         \centering
         \includegraphics[width=0.5\textwidth]{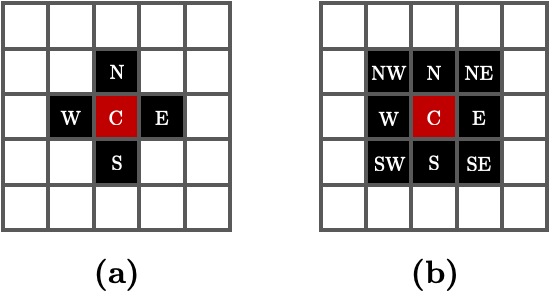}

 			\caption{A diagram of the cellular neighbourhoods defined by (a) Von Neumann and (b) Moore on a static grid lattice.}
        \label{fig:grid_neighbourhoods}
\end{figure}
In order to investigate the influence of the connection topology on emergence and stability of bilayer laminar patterning using a lateral-inhibition model, we used a variety of 2D and 3D geometries in both our analytical and numerical arguments. \autoref{fig:all connectivity diagrams} illustrates  the types of lattice geometry used, and a summary of all geometries can be found in \autoref{tab:geometry sum}.

 \begin{figure}[h!]
         \centering
         \includegraphics[width=0.8\textwidth]{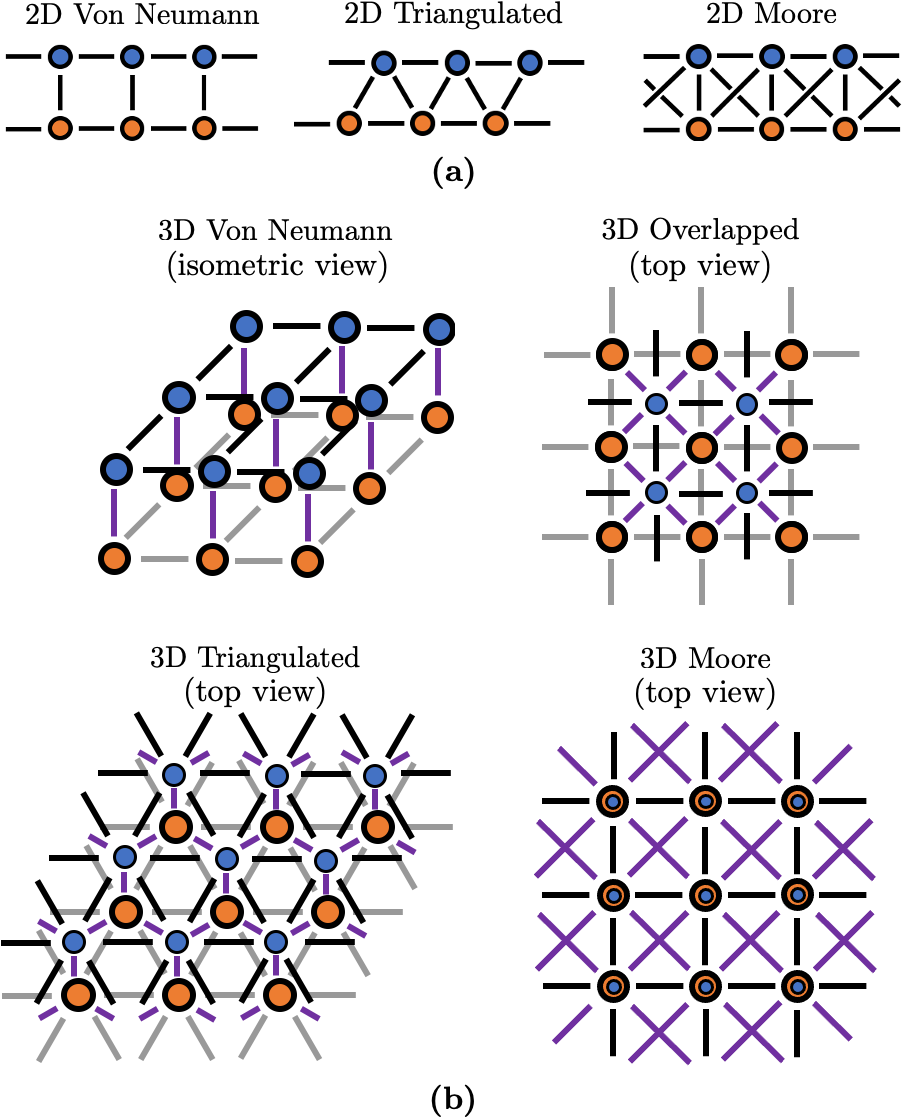}

 			\caption{Connectivity diagrams of the bilayer mammary organoid. Basal cells are shaded blue and luminal cells are shaded orange. (a) Representative diagrams of 2D geometries studied, solid black lines correspond to cellular connections. (b) Schematics for 3D lattices, grey lines correspond to connections between luminal cells, black lines are connections between basal cells and purple lines represent connections between the layers.}
        \label{fig:all connectivity diagrams}
\end{figure}

\begin{table}[h!]
\begin{center}
\resizebox{\linewidth}{!}{%
\begin{tabular}{ |c|c|c|c| }
 \hline
Lattice type & Connectivity radius, $\rho_{c}$ & Cellular connectivity, $N_{c}$ & Cell-type ratio, $R_{\tau}$   \\ 
\hline
2D Von Neumann (2DVN)  & 1 & 3 & 2 \\
2D Triangulated (2DT)& 1 & 4 & 1 \\
2D Moore (2DM)  & $\sqrt{2}$ & 5 & 2/3 \\
\hline
3D Von Neumann (3DVM)  & 1 & 5 & 4 \\
3D Overlapped (3DO1) & 1 & 8 & 1 \\
3D Triangulated (3DT)& 1 & 9 & 2 \\
3D Overlapped (3DO2)& $\sqrt{2}$ & 12 & 2 \\
3D Moore (3DM)  & $\sqrt{2}$ & 13 & 8/5 \\
\hline

\end{tabular}}
\caption{A summary of the lattice geometries in 2D and 3D that can be found in \autoref{fig:all connectivity diagrams} outlining the cellular neighbourhoods.}
\label{tab:geometry sum}
\end{center}
\end{table}

 \newpage

\subsection{Graph partitioning and network stability}\label{sec:graph_theory}

This study extends the work of Ferreira et al. (2013) where the symmetry of the cellular domain represented by a graph $\mathcal{G} = \mathcal{G}(V,E)$ was exploited to develop theoretical conditions on cellular connectivity for the existence and stability of inhomogeneous steady states in lateral-inhibition ODE models \cite{RufinoFerreira2013}. These methods were employed by considering contrasting pattern states of cells as partitions of the graph. A graph partition $P$ is a reduction of $\cG$ to a smaller graph, $\cG_{P}$, by partitioning the vertices, $v \in V$ into disjoint sets \cite{bollobas2013modern}. For example, each cell in $\cG$, represented by a vertex $v\in V$, can be collected into sets that converge to the same biochemical states, thus producing subsets of $V$ defining the graph partition $P$ (\autoref{fig:graphs}a). We say that the reduced graph $\cG_{P}$ is \textit{regular} if every vertex in $\cG_{P}$ has the same number of adjacent vertices. In addition, we introduce two properties of graphs, \textit{equitable} and \textit{bipartite}, that we use to derive conditions for the existence of contrasting states between partitions. 
\begin{defn}[Equitable graphs \cite{Godsil1997}]
For a graph to be \textit{equitable}, the subsets of vertices $V_{i} \subset V$ defining the graph partitions are themselves regular graphs, $\mathcal{G}_{i} = \mathcal{G}_{i}\left( V_{i}, E_{i} \right)$, and for any vertex $v_{i} \in V_{i}$, then $v_{i}$ has the same number of connections to vertices $v_{j} \in V_{j}$, independent of the choice of $v_{i}$ (\autoref{fig:graphs}b).
\end{defn}
\begin{defn}[Bipartite graphs \cite{bollobas2013modern}]
A graph $\mathcal{G}_{i}$ is \textit{bipartite} if $\mathcal{G}_{i}$ can be constructed by the union of two disjoint sets of vertices (\autoref{fig:graphs}c).
\end{defn}
Throughout this study, we assume the following:
\begin{assumption}
Each partition $P_{i}$ of the graph $\mathcal{G}$ is \textit{equitable}. 
\end{assumption}

\begin{assumption}
Each partition $P_{i}$ representing a reduced graph $\mathcal{G}_{i}$ of the graph $\mathcal{G}$ is \textit{bipartite}. 
\end{assumption}
\noindent By applying Assumptions 3.1-3.2 to partitioned graphs, Ferreira et al. (2013) was able to provide a framework for existence of \textit{template} patterns defined by the partitions of the vertices. The application of graph partitioning allowed for the reduction of the \textit{scaled adjacency matrix} $(1/N_{w})\bm{W} \in \R ^{N \times N}_{\geq 0}$ to a \textit{quotient matrix} $(1/N_{w})\overline{\bm{W}} \in \R ^{r \times r}_{\geq 0}$, which represents the connectivity of the partitions as proportional values between representative cells from each partition, and $r$ is the number of partitions of the graph $\mathcal{G}$. However, the dimension reduction of connectivity assumes that each cell within the same partition behaves identically.

 \begin{figure}[h!]
         \centering
         \includegraphics[width=0.7\textwidth]{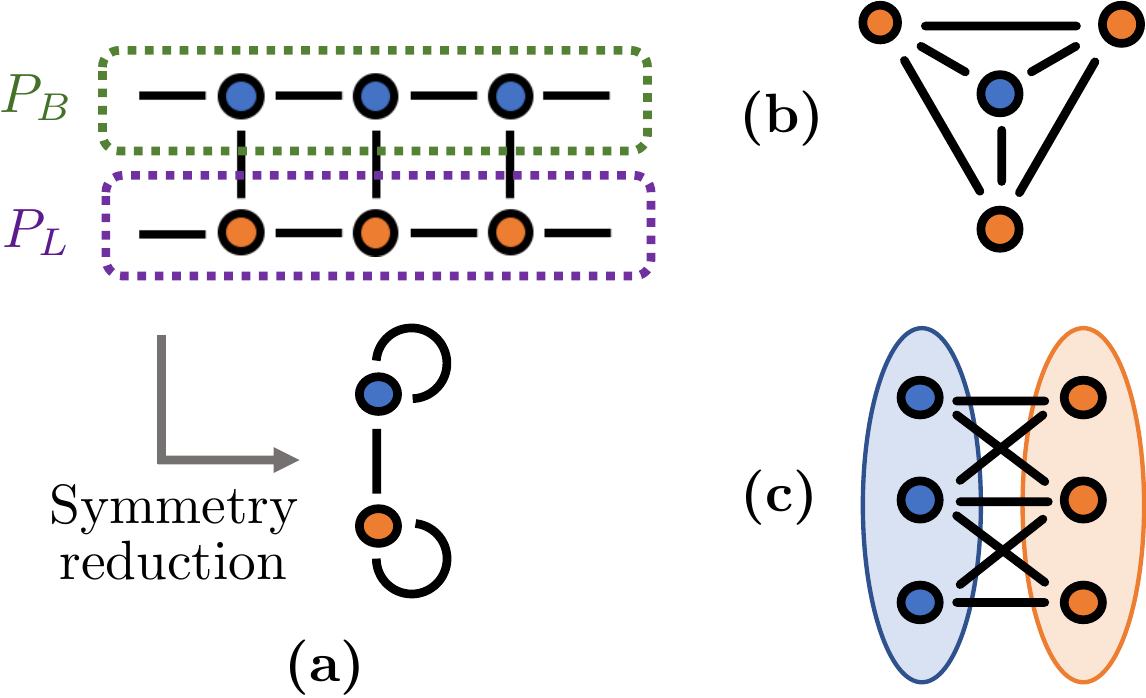}

 			\caption{Graph partitioning, symmetry reduction and example graphs. (a) A diagram representing the partitioning and symmetry reduction process for a bilayer of vertices. The graph defining the bilayer is partitioned into layers, such that basal cells lie in the partition $P_{B}$ and luminal cells lie in the partition $P_{L}$, that is $r = 2$, where $r$ is the dimension of the quotient space. By exploiting the symmetry of the graph, a reduction of vertices is made to consider only representative cells from each partition, $P_{B}$ and $P_{L}$. (b) An example graph with the equitable partitions property. The partitions are highlighted using colours, the diagram highlights that the blue node always has three connected orange nodes, and any orange node has only one blue node connected. (c) An example graph with the bipartite property. The full graph consisting of both blue and orange nodes can be decomposed into two disjoint sets of nodes highlighted by the shaded regions.}
        \label{fig:graphs}
\end{figure}

We consider each vertex representing a cell to contain a spatially discrete ODE system describing lateral-inhibition, these systems have previously been represented as single-input-single-output systems (SISO), a common representation of ODEs in control theory \cite{Arcak2013}. Formally, the SISO of a lateral feedback model has the form, for each cell $i$,
\begin{align}\label{SISO}
\dot{\bm{x}}_{i} &= \bm{f}\left(  \bm{x}_{i}, u_{i}   \right),\\
y_{i} &= h\left(\bm{x}_{i} \right),
\end{align}
where $\bm{x}_{i} \in \bm{X}$ is a vector of reactants (e.g. Notch-Delta), $u_{i} \in U$ is the input value to each cell, determined by the discrete spatial operator $\langle  \cdot \rangle$ and $y_{i} \in Y$ is the output of cell $i$ to its connected neighbours. The function $\bm{f}$ defines the nonlinear dynamics of the feedback model, and $h$ is the function defining the relationship between the ODEs and the output of the cell. It is assumed that both $\bm{f}$ and $h$ are continuously differentiable. Furthermore, we introduce the map called the characteristic transfer function $T: U \rightarrow Y$, which defines the transfer of flow of the dynamical system, 
\begin{equation}
T(\cdot) := h\left( S( \cdot ) \right),
\end{equation}
where $S: U \rightarrow X$ is a function mapping the information from connected cells to the dynamical system (3.9). It is assumed that $T$ is positive and bounded, and characteristically, $T$ is a decreasing function for lateral-inhibition and increasing for lateral-induction. For the nonlinear dynamics required to produce patterning via lateral-inhibition mechanisms, the characteristic transfer function, $T$, is generally algebraically intractable as it is constructed by the composition of nonlinear functions that define the systems dynamics. By linearising the SISO system (3.9-3.10) near points of interest in the spaces $\bm{X}, U$ and $Y$ allows for tractable analysis to investigate the stability of the SISO system (3.9-3.10). \par
A general linear SISO system has the form,
\begin{align}
\dot{\bm{x}}_{i} &= \bm{Ax}_{i} + \bm{B}u_{i},\\
y_{i} &= \bm{Cx}_{i},
\end{align}
for $\bm{A} \in \R^{n \times n}$, $\bm{B} \in \R^{n \times 1}$ and $\bm{C} \in \R^{1 \times n}$, where $n = \dim\left(\bm{x}_{i} \right)$. The derivative transfer function $T^{\prime}$ of a general SISO is analogous to the transfer matrix ${G}$ in linear SISO systems, and has the form,
\begin{equation}
{G}(s_{i}) = \bm{C} \left(s_{i} \mathbb{1}_{n} - \bm{A}\right)^{-1}\bm{B},
\end{equation}
as derived in \cite{Weiss1994}. It can be shown by linearising the SISO system for lateral feedback models (equations (3.9-3.10)), the dynamics of the transfer function $T^{\prime}$ can be approximated to linear form near steady state. That is,
\begin{equation}\label{eqn:transfer_der}
T^{\prime}(u_{s}) = - \left( \frac{\partial h}{\partial \bm{x}} \right) \left(  \frac{\partial \bm{f}}{\partial \bm{x}}    \right)^{-1} \left( \frac{\partial \bm{f}}{\partial u} \right)\biggr\rvert_{\bm{x} = \bm{x}_{s}}  = -\bm{C}\bm{A}^{-1}\bm{B} = G(0),
\end{equation}
as demonstrated in \cite{Arcak2013}, where $u_{s}$ is a steady state input for the SISO system for lateral feedback models.

A key property of SISO lateral-inhibition models is monotonicity. Monotone systems preserve the order of trajectories within respective nonempty subsets of \textit{Banach} spaces \cite{Angeli2003}. The trajectory spaces $K$ we consider have the following properties (\autoref{fig:cones}):
\begin{enumerate}
\item $K$ is a cone in Euclidean space, that is, $\alpha K \subset K $ for $\alpha \in \R_{\geq 0}$.
\item $K$ is convex, $K+K \subset K$.
\item $K$ is pointed, namely, $K  \cap \left( -K \right)  = \{0\}$.
\end{enumerate}

\noindent Given a cone $K$, we define partial ordering of elements via ``$\preceq$'' such that $x \preceq \hat{x}$ means that $\hat{x} - x \in K$  \cite{Arcak2013}. By embedding the trajectory spaces $\bm{X}, U$ and  $Y$, into cones $K^{\bm{X}}, K^{U}$ and $K^{Y}$, we formally define monotonicity of SISO systems (3.9-3.10). 

\begin{figure}[h!]
         \centering
         \includegraphics[width=0.4\textwidth]{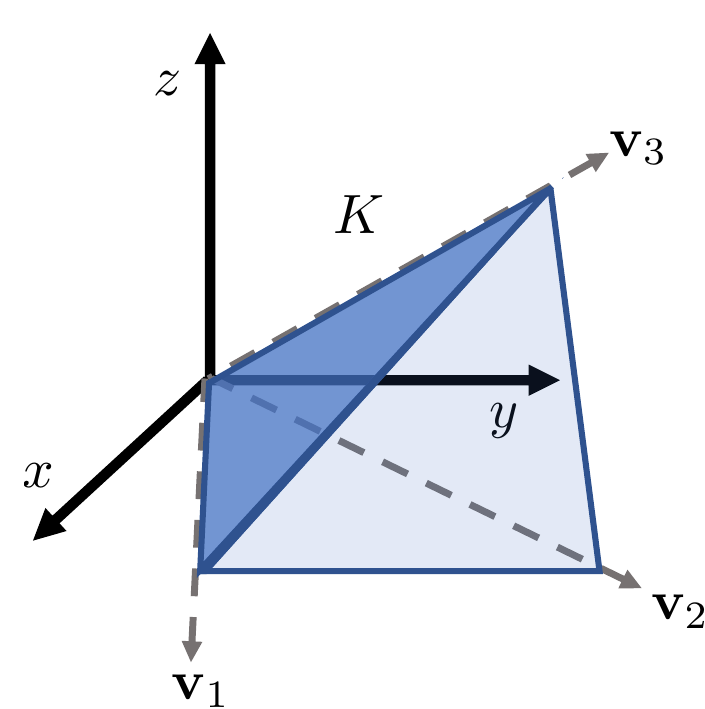}

 			\caption{An example pointed, convex cone $K$ in Euclidean space, where $K = \{ \bm{x} \in \R^{3} \, | \, \bm{x} =a_{1}\mathbf{v}_{1}+ a_{2}\mathbf{v}_{2}+ a_{3}\mathbf{v}_{3} \, \, \forall a_{i} \in \R_{\geq 0}  \}. $}
        \label{fig:cones}
\end{figure}

\begin{defn}[Monotone SISO system \cite{Angeli2003}]\label{def:mono}
Given the cones $K^{U},K^{Y},K^{\bm{X}}$ for the input, output and state spaces, the input-output ODE model $\dot{\bm{x}}_{i} = \bm{f}(\bm{x},u_{i}), \, y_{i} = h(\bm{x}_{i})$ is said to be monotone if $\bm{x}_{i}(0) \preceq \hat{\bm{x}}_{i}(0)$ and $u_{i}(t) \preceq \hat{u}_{i}(t)$ for all $t \geq 0$ imply that the resulting solutions satisfy $\bm{x}_{i}(t) \preceq \hat{\bm{x}}_{i}(t)$ for all $t \geq 0$, and the output map is such that $\bm{x}_{i} \preceq \hat{\bm{x}}_{i}$ implies $h(\bm{x}_{i}) \preceq h\left(\hat{\bm{x}}_{i}\right) $.
\end{defn} 
\noindent The monotonicity of SISO systems (3.9-3.10) have previously be used to investigate the stability of component-wise steady states \cite{Angeli2003,Arcak2013}.\par
In control theory, the stability of SISO systems (3.9-3.10) can be assessed by analysing the transition of inputs and output between components of the connected system. A particular measure of a connected system is the $\mathcal{L}_{2}$-gain, which is a nonnegative quantity which describes the response of a system to an input. We first provide a general definition of a $\mathcal{L}_{p}$-gain, $\hat{q}$.

\begin{defn}[$\mathcal{L}_{p}$-gain of a SISO system \cite{DerSchaft1996}]\label{def:L2}
The $\mathcal{L}_{p}$-gain, $\hat{q}_{i} > 0$,  of a SISO system (3.9-3.10) is defined by
\begin{equation}
\hat{q}_{i} = \sup_{T>0} \lp \frac{||u_{i}||_{p}}{||y_{i}||_{p}} \rp
\end{equation}
for all $y_{i}$ and $u_{i}$ for $i = 1,...,N$,  and T > 0 denotes the truncation of the Hilbert spaces for the input and outputs of the system, $\mathcal{L}_{p} \lp U \rp$ and $\mathcal{L}_{p} \lp Y \rp$, respectively.
\end{defn}

For monotone systems with asymptotically stable steady states, the  $\mathcal{L}_{2}$-gain, $q_{i}$ for $ i = 1,...,N$ (and the  $\mathcal{L}_{\infty}$-gain by Young’s inequality and gain relations \cite{doyle2013feedback}) are identical to the steady state gain (dc-gain), i.e. the ratio of the steady state output and steady input. That is,
\begin{equation}\label{eqn:l2_t}
q_{i} = ||G(0)||_{\infty} = |G(0)| = |-\bm{C}\bm{A}^{-1}\bm{B}|.
\end{equation}
For nonlinear systems, linearising around a desired point in the state space yields the same result \cite{Angeli2004}. Hence, the $\mathcal{L}_{2}$-gain (and $\mathcal{L}_{\infty}$-gain) can be obtained by $q_{i} = |T^{\prime}(u_{s,i})|$, where $u_{s,i}$ is the steady state of the input of cell $i$. Therefore, by demonstrating a nonlinear SISO system is monotone, we have a convenient procedure to compute the $\mathcal{L}_{2}$-gains for each cell in the system. \par 
The $\mathcal{L}_{2}$-gains of an interconnected SISO systems are particularly useful for understanding the stability of the feedback between the connected components. The Small Gain theorem yields a sufficient bound on two interconnected components for the global stability of feedback. 

\begin{theorem}[Small Gain theorem  \cite{DerSchaft1996}]\label{thm:small-gain}
For all bounded inputs, a SISO system (3.9-3.10) of two interconnected components, $c_{1}$ and $c_{2}$, in a closed-loop are locally asymptotically stable if $c_{1}$ and $c_{2}$ are independently stable and
\begin{equation}
q_{1}q_{2} < 1,
\end{equation}
where $q_{1}$ and $q_{2}$ are the $\mathcal{L}_{2}$-gains of $c_{1}$ and $c_{2}$, respectively.
\end{theorem}

Using the information we have outlined in this section, we apply the following theorems derived by Ferreira et al. (2013), that focused on lateral-inhibition models to bilayer geometries in 2D and 3D \cite{RufinoFerreira2013}. The first theorem we consider provides an inequality relating the characteristics of the SISO dynamics to the geometry of the symmetric cellular domain that yields the instability of the homogeneous steady state for all cells. 
\begin{theorem}[Thm 4 in \cite{RufinoFerreira2013}]\label{theorem: reduced}
Let $\pi$ be an equitable partition of the vertices of $\mathcal{G}$ such that the reduced graph $\mathcal{G}_{\pi}$ is bipartite. Let $\overline{\lambda}_{r}$ the smallest eigenvalue of reduced quotient matrix $(1/N_{w})\overline{\bm{W}}$. If the output characteristic function, $T$, is positive, bounded and decreasing, and if for the homogeneous input steady state, $u^{*}$, we have, 
\begin{equation}\label{part_condition}
|T^{\prime}(u^{*})|\overline{\lambda}_{r} < -1,
\end{equation}
then there exists heterogeneous steady states in the representative partitions.
\end{theorem}

The second main result of Ferreira et al. (2013) provided conditions for the stability of heterogeneous steady states via $\mathcal{L}_{2}$-gain conditions in symmetric cellular domains. By exploiting the symmetry of the domain, and therefore assuming each cell within the same partition behaves identically, the $\mathcal{L}_{2}$-gains for those cells will also be identical. Let  $\overline{\bm{Q}} = \text{diag} \{q_{1},...,q_{r}\}$ represent the $\mathcal{L}_{2}$-gains from each representative cell in each patterning partition. Before stating the result, we introduce some notation for the spectral properties of the graph. Let $\bm{A}$ be a square matrix with the set of eigenvalues $\sigma \left( \bm{A} \right)$. Then the spectral radius of $\bm{A}$ is defined by 
\begin{equation}
\rho \left( \bm{A} \right) = \max \{ |\lambda| : \lambda \in \sigma \left( \bm{A} \right) \}.
\end{equation}
Using the spectral properties of the connected graph defined by the cellular domain, the stability criterion for the heterogeneous steady states follows from a theorem proven in \cite{RufinoFerreira2013} which is stated as follows.
\begin{theorem}[Thm 10 in \cite{RufinoFerreira2013}]\label{thm: stability}
Consider the network as defined in \autoref{theorem: reduced}. The steady state pattern defined by heterogeneous steady states are locally asymptotically stable if 
\begin{equation}\label{spectral bound}
\rho\left( (1/N_{w})\overline{\bm{W}} \, \overline{\bm{Q}} \right) < 1,
\end{equation}
where $(1/N_{w})\overline{\bm{W}}$ is the \textit{reduced quotient matrix} and $\rho(\cdot)$ represents the spectral radius.
\end{theorem}

\noindent \autoref{theorem: reduced} and \autoref{thm: stability} form the basis of our analytical results, which focus on how cellular connectivity in stratified bilayer geometries can influence the polarisation requirements to maintain laminar pattern formation.

\section{Theoretical results}\label{sec:theoretical_results}
We now focus on deriving algebraic conditions on the cell-type dependent signal strength parameters $w_{1}, w_{2}$ that yield bilayer laminar patterning in lateral-inhibition models. By leveraging the pattern \textit{template} methods outlined in \autoref{sec:graph_theory} (\autoref{fig:graphs}a), we significantly reduce the complexity involved juxtacrine pattern analysis in multicellular systems to investigate the role of geometry and cellular connectivity in signal stratification.

\subsection[Existence of bilayer patterns]{Conditions on $w_{1}$-$w_{2}$ defining the existence of bilayer laminar pattern formation}\label{sec:exist}


By applying \autoref{theorem: reduced} to the refined geometry of a bilayer periodic lattice yields the following restriction on the parameters $w_{1}$ and $w_{2}$ for the existence of laminar patterning.

\begin{theorem}\label{prop: bound}
Let $\mathcal{G}$ be an undirected, cyclic, connected graph with equitable partitions $P_{L}$ and $P_{B}$, such that the reduced graphs $\mathcal{G}_{L}$ and $\mathcal{G}_{B}$ are bipartite, where $\mathcal{G}_{L}$ contains the set of vertices representing luminal cells and $\mathcal{G}_{B}$ contains the set of vertices representing basal cells. If the output characteristic function, $T$, is positive, bounded and decreasing, then there exists heterogeneous steady states between partitions $P_{L}$ and $P_{B}$ if,
\begin{equation}\label{eqn: existence_bound}
w_{1} <\left(  \frac{|T^{\prime}(u^{*})| - 1 }{|T^{\prime}(u^{*})|+1}  \right) \frac{w_{2}}{R_{\tau}},
\end{equation}
provided that $n_{1}w_{1} < n_{2}w_{2}$.
\end{theorem}

\begin{proof}
Let $(1/N_{w})\bm{W}$ be the general \textit{scaled adjacency matrix} for a bilayer of cells consisting of stratified cell-types, as seen in \autoref{fig:all connectivity diagrams}. Due to the symmetry of the connected graph, $\bm{W}$ is of the block form, as in equation\autoref{eqn:block_w}. Let $N \in 2\mathbb{N}$ be the total number of cells in the bilayer, then for each row $i \in [1,N/2]$, there is exactly $n_{1}$ and $n_{2}$ non-zero entries in $\bm{W}_{1}$ and $\bm{W}_{2}$, respectively. Therefore, the \textit{reduced quotient matrix} $(1/N_{w}) \overline{\bm{W}}$, depicting the quotient graph in \autoref{fig:graphs}a, which can be regarded as proportional connectivity, has the following form,
\begin{equation}
(1/N_{w}) \overline{\bm{W}} =\frac{1}{N_{w}} \left[
  \begin{array}{cc}
  n_{1}w_{1} & n_{2}w_{2}\\
    n_{2}w_{2} &  n_{1}w_{1}
   \end{array}
\right].
\end{equation}
The eigenvalues of $(1/N_{w}) \overline{\bm{W}}$ can be directly determined as,
\begin{equation}
\lambda_{1} = 1 \qquad \mbox{and} \qquad \lambda_{2} = \frac{n_{1}w_{1} - n_{2}w_{2}}{N_{w}},
\end{equation} 
and therefore $\lambda_{2} < \lambda_{1}$ for any connectivity graph in bilayer geometries. Hence, by the assumption $n_{1}w_{1} < n_{2}w_{2}$, we deduce that $\lambda_{2} = \overline{\lambda}_{r}$ in \autoref{theorem: reduced}. Applying \autoref{theorem: reduced} to the bilayer geometry, we substitute $\lambda_{2} $ into inequality\autoref{part_condition},
\begin{equation}
|T^{\prime}(u^{*})|\left(  \frac{n_{1}w_{1} - n_{2}w_{2}}{n_{1}w_{1} + n_{2}w_{2}}    \right) < -1,
\end{equation}
which can be rearranged to yield inequality\autoref{eqn: existence_bound}.
\end{proof}

Inequality\autoref{eqn: existence_bound} bounds the cell-type dependent signal strength and highlights the influence of cellular connectivity on pattern formation via the $R_{\tau}$ value. As $n_{1}$ increases, this implies $R_{\tau}$ increases, thus pattern existence is restricted by inequality \autoref{eqn: existence_bound}, requiring greater signal polarisation to cells of a different type (the converse is also true).

\subsection[Stability of bilayer patterns]{Conditions on $w_{1}$-$w_{2}$ defining the stability of bilayer laminar pattern formation}\label{sec:stab}

By applying \autoref{thm: stability} to the refined geometry of a bilayer periodic lattice yields the following restriction on the parameters $w_{1}$ and $w_{2}$, via the SISO system $\mathcal{L}_{2}$-gains, where these restrictions can be stated explicitly in particular cases.  That is,  the $\mathcal{L}_{2}$-gains, $q_{i}$ are sufficiently bounded for the stability of heterogeneous steady states $\bm{x}_{B}$ and $\bm{x}_{L}$ which define the contrasting signal expression in each partition, $P_{B}$ and $P_{L}$, respectively.

\begin{theorem}\label{coro_stab}
Let $\mathcal{G}$ be an undirected, cyclic, connected graph with equitable partitions, $P_{L}$ and $P_{B}$, such that the reduced graphs $\mathcal{G}_{L}$ and $\mathcal{G}_{B}$ are bipartite,  where $\mathcal{G}_{L}$ contains the set of vertices representing luminal cells and $\mathcal{G}_{B}$ contains the set of vertices representing basal cells. Then if the $\mathcal{L}_{2}$-gains of the representative cells, $q_{B}$ and $q_{L}$,  satisfy
\begin{equation}\label{eqn_full_stab_bound}
\frac{n_{1}w_{1} \lp q_{B}+q_{L} \rp + \sqrt{ \lp  n_{1}w_{1} \lp q_{B} - q_{L}\rp   \rp^{2} + \lp2 n_{2}w_{2} \sqrt{q_{B} q_{L} }    \rp^{2} }}{2 \lp n_{1}w_{1} + n_{2}w_{2} \rp} <1
\end{equation}
then the heterogeneous steady states $\bm{x}_{B}$ and $\bm{x}_{L}$ associated with $q_{B}$ and $q_{L}$ are locally asymptotically stable.  Moreover,  if 
\begin{equation}\label{eqn:linear_condition}
w_{1} < \lp\frac{2}{q_{B}+q_{L} - 2} \rp \frac{w_{2}}{R_{\tau}}
\end{equation}
then the local asymptotic stability of the heterogeneous steady states can be ensured by the following bounds on $w_{1}$:

\begin{equation}\label{eqn:explicit_stab}
w_{1} < \left\{
	\begin{array}{ll}
		S_{+}\lp q_{B},q_{L} \rp \frac{w_{2}}{R_{\tau}} & \mbox{if } 1 +q_{B}q_{L} < q_{B} + q_{L}, \\
		\lp \frac{q_{B}q_{L}}{q_{B}q_{L} - 1} \rp \frac{w_{2}}{R_{\tau}} & \mbox{if }  1 +q_{B}q_{L} = q_{B} + q_{L},  \\
		S_{-}\lp q_{B},q_{L} \rp \frac{w_{2}}{R_{\tau}} & \mbox{if } 1 +q_{B}q_{L} > q_{B} + q_{L},
	\end{array}
\right.
\end{equation}
where 
\begin{equation}
S_{\pm} \lp q_{B}, q_{L} \rp =  \frac{2 - \lp  q_{B} +q_{L} \rp \pm \sqrt{ \lp 4q_{L} -3  \rp q_{B}^{2} + \lp 4q_{L}^{2} -10 q_{L} + 4 \rp q_{B} +\lp 4 - 3q_{L} \rp q_{L}       } }  {2 \lp q_{B}q_{L} -q_{B} -q_{L} +1 \rp} .
\end{equation}

\end{theorem}

\begin{proof}

Let $(1/N_{w})\overline{\bm{W}}$ to be the \textit{reduced quotient matrix} for a cell-type stratified bilayer of cells as defined in the proof of \autoref{prop: bound}. Let $a = n_{1}w_{1}/N_{w}$ and $b = n_{2}w_{2}/N_{w}$ for brevity. Therefore we have,
\begin{equation}
(1/N_{w})\overline{\bm{W}}\, \overline{\bm{Q}} = \left[
  \begin{array}{cc}
  aq_{B} & bq_{L}\\
    bq_{B} &  aq_{L}
   \end{array}
\right].
\end{equation}
As $(1/N_{w})\overline{\bm{W}} \, \overline{\bm{Q}}$ is a nonnegative matrix, by Perron's theorem \cite{Chang2008}, $\rho\left( (1/N_{w})\overline{\bm{W}} \, \overline{\bm{Q}} \right)$ is an eigenvalue of $(1/N_{w})\overline{\bm{W}}\, \overline{\bm{Q}}$, which is real due to the positivity of the matrix. By directly solving for the eigenvalues of $(1/N_{w})\overline{\bm{W}} \, \overline{\bm{Q}}$, we have,
\begin{equation}
\rho\left((1/N_{w})\overline{\bm{W}} \, \overline{\bm{Q}}   \right) = \frac{a\left( q_{B} + q_{L} \right)}{2} + \frac{\sqrt{a^{2} \left(q_{B} - q_{L}  \right)^{2} + 4 b^{2}q_{B}q_{L}  }}{2},
\end{equation}
and so by applying \autoref{thm: stability} yields inequality\autoref{eqn_full_stab_bound} for the local asymptotic stability of the heterogeneous states.  If inequality\autoref{eqn:linear_condition} holds,  then we have that
\begin{equation}\label{eqn:pos_bound}
\sqrt{a^{2} \left(q_{B} - q_{L}  \right)^{2} + 4 b^{2}q_{B}q_{L}  } < 2 - a \lp q_{B}+q_{L} \rp
\end{equation}
is positive,  which preserves the sign when taking the squares of both sides of inequality\autoref{eqn:pos_bound}.  Rearranging the bound yields,
\begin{equation}\label{inequal:bound}
a^{2}\left(q_{B} - q_{L}\right)^{2} + 4b^{2} q_{B}q_{L} < \left( 2- a\left( q_{B} + q_{L} \right)       \right)^{2}.
\end{equation}
Expanding the expressions in equation\autoref{inequal:bound} and rearranging gives,
\begin{equation}\label{ineq:final_bound}
q_{B}q_{L} \left( b^{2} - a^{2} \right) + a \left( q_{B} + q_{L} \right) < 1.
\end{equation}
Substituting $a$ and $b$ back into equation\autoref{ineq:final_bound} and yields a quadratic form in $w_{1}$ (and $w_{2}$)
\begin{equation}
n_{1}^{2}\lp q_{B} +q_{L} - q_{B}q_{L}-1 \rp w_{1}^{2} + n_{1}n_{2}w_{2} \lp q_{B} + q_{L} - 2 \rp w_{1} + \lp q_{B} + q_{L} - 1 \rp \lp n_{2} w_{2} \rp ^{2} <0
\end{equation} 
which can be directly solved provided the know information of the sign of the parabola.
\end{proof}

\begin{corollary}
If the homogeneous steady state $u^{*}$ of a monotone SISO system\autoref{SISO} yields $|T^{\prime} \lp  u^{*}\rp | \geq 3$,  then inequality\autoref{eqn:linear_condition} is always satisfied. 
\end{corollary}

\begin{proof}
Without loss of generality,  we have that $q_{L}<|T^{\prime} \lp  u^{*}\rp |<q_{B}$ as the contrasting input states $u_{L}$ and $u_{B}$ will diverge from $u^{*}$ in opposing directions.  From inequality\autoref{eqn: existence_bound} we know that $1 < q_{B}$ must hold as $w_{1} \in \R_{\geq 0}$.  Therefore,  if we assume that the homogeneous steady state of the monotone SISO system\autoref{SISO} is unstable,  which is required for laminar patterning and so inequality\autoref{eqn: existence_bound} is satisfied.  Comparing inequalities\autoref{eqn: existence_bound} and\autoref{eqn:linear_condition},  we have that inequality\autoref{eqn:linear_condition} holds when
\begin{equation}\label{eqn:pb_t}
\frac{1}{q_{B}-1} < \frac{2}{q_{B}+q_{L}-2} < \frac{|T^{\prime} \lp  u^{*}\rp | - 1}{|T^{\prime} \lp  u^{*}\rp | + 1},
\end{equation}
where the left-most term follows from $q_{L}<q_{B}$.  Rearranging inequality\autoref{eqn:pb_t} yields
\begin{equation}\label{eqn:approx_pb_t}
\frac{2 |T^{\prime} \lp  u^{*}\rp | }{|T^{\prime} \lp  u^{*}\rp | -1} <q_{B},
\end{equation}
then applying our assumption $|T^{\prime} \lp  u^{*}\rp | < q_{B}$,  inequality\autoref{eqn:approx_pb_t} can be solved directly, providing the minimum value of $ |T^{\prime} \lp  u^{*}\rp | = 3$ to satisfy inequality\autoref{eqn:pb_t}.
\end{proof}

Inequality\autoref{eqn_full_stab_bound} outlines the relationship between cellular connectivity and signal protein feedback that is required to be balanced to ensure the maintenance of pattern formation in bilayer static geometries. However, we note that the $\mathcal{L}_{2}$-gains are dependent on the geometry, as they are a function of the input value defined by the discrete spatial operator $\langle \cdot \rangle$, see \autoref{def:L2}. Thus, inequality\autoref{eqn_full_stab_bound} cannot determine explicit conditions for the relationship between geometry and feedback model, as in the existence inequality\autoref{eqn: existence_bound}.\par  
In addition, inequality\autoref{ineq:final_bound} and therefore inequality\autoref{eqn:explicit_stab} describes a relaxation of the Small Gain theorem for closed-loop system, commonly used in control theory applications \cite{haddad2011nonlinear}. To demonstrate this relaxation of the Small Gain theorem, w.l.o.g. assume that $q_{L} < q_{B}$, as we expect the partitions $P_{B}$ and $P_{L}$ of $\mathcal{G}$ to obtain contrasting solution states. In this case, inequality\autoref{ineq:final_bound} is bounded above,
\begin{equation}
q_{B}q_{L} \left( b^{2} - a^{2} \right) + a \left( q_{B} + q_{L} \right)  < q_{B} \lp \lp b^{2} - a^{2} \rp q_{B} + 2a \rp,
\end{equation}
where $a = n_{1}w_{1}/N_{w}$ and $b = n_{2}w_{2}/N_{w}$ as in \autoref{coro_stab} and we impose that $b>a$ from \autoref{prop: bound}. Therefore if
\begin{equation}\label{eqn:small_gain_equiv}
q_{B}\left( \left(b^{2} - a^{2} \right)q_{B} + 2a \right) < 1
\end{equation}
holds, then inequality\autoref{ineq:final_bound} must also be satisfied. Solving inequality\autoref{eqn:small_gain_equiv} for $q_{B}$ implies if $q_{B}<1$ then the dynamical system is locally asymptotically stable. Moreover, if $q_{B}<1$ then $q_{B}^{2}<1$, which implies $q_{B}q_{L}<1$, thus satisfying the Small Gain theorem (\autoref{thm:small-gain}). As a special case of \autoref{coro_stab}, if each cell in the cellular domain has no adjacent cell of the same type, namely $n_{1} = 0$, then inequality\autoref{ineq:final_bound} is equivalent to the Small Gain theorem, as demonstrated previously for checkerboard patterns using lateral-inhibition models (\autoref{fig:checkerboard_domains}) \cite{RufinoFerreira2013}.


\section[Application]{Application: Notch-Delta pattern formation in the mammary organoid}
To illustrate the results of sections 4.1-4.2, we use the Notch-Delta lateral-inhibition model outlined in Section 3.1 to impose signalling polarity conditions for luminal and basal cells in the mammary organoid to achieve the experimentally observed laminar pattern formation of Notch (and therefore Delta) in a bilayer of cells (Figures 3b-3d). In addition, we check the validity of theoretical bounds by conducting fixed lattice simulations for bilayer cyclic structures using the ODE system (3.1-3.2). Furthermore, using the analytic and numerical polarisation conditions on $w_{1}$ and $w_{2}$, we employ cell-based modelling to highlight the ability of the family of static lattice structures (see \autoref{tab:geometry sum}) to gain insight into signal polarity conditions in dynamic cellular domains.

\subsection{Methods for static lattice simulations}\label{sec:static_methods}
The 2D fixed lattice geometries were considered as a 6 cell system, split equally as luminal and basal cells as demonstrated in \autoref{fig:all connectivity diagrams}a. This cyclic geometry generates a system of 12 ODEs that were coupled via the \textit{scaled adjacency matrix} $(1/N_{w})\bm{W}$ as previously discussed in Section 3.2. Similarly, 3D fixed geometries were treated as a cyclic 18 cell system, and therefore producing a system of 36 ODEs. For both 2D and 3D geometries, the ODE systems were solved numerically using the \texttt{ode45} function in Matlab 2019b. The simulations were solved over a period of 100 hours to allow for the system to reach a steady state.\par 

To determine if the luminal and basal layers have converged to contrasting states of Notch-Delta expression, the mean value of Delta expression was taken from each layer of cells. Explicitly, let $d_{j}$ denote the mean final Delta values in each layer of cells, such that,
\begin{equation}
d_{j}  = \frac{2}{\mathcal{N}} \sum_{i}^{\mathcal{N}}D_{i,j},
\end{equation}
where $\mathcal{N}$ is the total number of cells in system. The difference $\Delta d = |d_{1} - d_{2}|$ indicates the existence of laminar bilayer pattern formation. We considered the system to have achieved a laminar bilayer pattern if $\Delta d$ was greater than a prescribed tolerance, $\delta>0$. \par
The static simulation parameter sweeps for $w_{1}$ and $w_{2}$ where conducted over a discretised $150 \times 150$ regular grid lattice for $w_{1} \in (0,0.25]$ and $w_{2} \in (0,1]$, resulting in 22500 simulations per static geometry. In all static lattice simulations, we choose $a = 0.01$, $b = 100$, $\mu_{N} = \mu_{D} = 1$, $h = 1$ and $k=2$ as parameter values for the NDM (3.1-3.2), as previously defined \cite{Collier1996}.

\subsection{Methods for lattice free simulations using a cell-based model}\label{sec:abm_methods}

Cell-based simulations were carried out using Chaste v2019.1 (Cancer, Heart and Soft Tissue Environment) \cite{Mirams2013}, where the Overlapping Spheres (OS) framework was used to enable seamless transition between 2D and 3D geometries. In addition, it has been previously demonstrated that OS models are highly applicable to study short ranged signal-reaction networks in cellular systems due to the mechanical methods used to define cellular contact \cite{Osborne2017}. Cells are connected by a mechanical force which is proportional to the region of overlap of spheres defined around each cellular node, see \autoref{fig:ABM_methods}. Here, we used the OS force model as defined in \cite{Atwell2015}, where, the displacement of two nodes representing a cell centres is represented by the vector $\bm{r}_{ij} = \bm{r}_{i} - \bm{r}_{j}$ and the force between the cells is defined by,
\begin{equation}
\bm{F}_{ij}(t) = 
		\begin{cases}
						    \eta_{ij}s_{ij}(t)\hat{\bm{r}}_{ij}(t) \log{\left(  1+  \frac{||\bm{r}_{ij}(t)|| -s_{ij}(t)}{s_{ij}(t)}      \right)}, &  \text{for}\quad ||\bm{r}_{ij}(t)|| < s_{ij}(t), \\
      						\eta_{ij} \left(  ||\bm{r}_{ij}(t)|| -s_{ij}(t)    \right)  \hat{\bm{r}}_{ij}(t) \exp{ \left( -k_{c} \frac{||\bm{r}_{ij}(t)|| -s_{ij}(t)}{s_{ij}(t)}         \right)},   &  \text{for}\quad  s_{ij}(t) \leq ||\bm{r}_{ij}(t)|| < r_{\text{max}}, \\
     						0, & \text{for} \quad ||\bm{r}_{ij}(t)|| > r_{\text{max}},
    \end{cases} 
\end{equation}
where $ \eta_{ij},s_{ij}(t)>0$ are the spring constant and rest length between cells $i$ and $j$. $\hat{\bm{r}}_{ij}(t)$ corresponds to the unit vector of $\bm{r}_{ij}(t)$ and $k_{c}$ defines the decay of force between the cells. Upon cellular division, the rest length $s_{ij}(t)$ of both parent and daughter cells are set to $s_{ij}^{\text{div}} = s_{ij}(t)/2$ and will tend back to $s_{ij}(t)$ in finite time as the cell grows.
In all simulations, random motion was introduced to each cell to stimulate a dynamic cellular domain. The random motion was implemented by an additional force acting on each cell node at each timestep, 
\begin{equation}
\bm{F}^{\text{rand}} = \sqrt{\frac{2\xi}{\Delta t}} \bm{\nu},
\end{equation}
where $\xi$ is a constant defining the size of random perturbation, $\bm{\nu}$ is a vector of samples from a standard multivariate normal distribution and $\Delta t$ the timestep of the simulation, as previously defined \cite{Osborne2017}. The resultant force acting on cell $i$ is defined by,
\begin{equation}
\bm{F}_{i}^{\text{res}}(t) = \bm{F}_{i}^{\text{rand}} + \sum_{j}^{\mathcal{N}_{i}} \bm{F}_{ij}(t),
\end{equation}
for $\mathcal{N}_{i}$ is the number of cells within the cut-off distance, $r_{\text{max}}$. Using this resultant force acting upon each cell, we relate this to cellular movement using the assumption that the inertia terms are small in  comparison to the dissipative terms acting upon the cell. This is because both \textit{in vivo} and \textit{in vitro} cells move in dissipative environments with small Reynolds number \cite{purcell1977life}, thus the position of each cell is governed in the Aristotelian regime, such that the velocity of a cell is proportional the force acting on it. Namely, the spatial dynamics of each cell is determined by,
\begin{equation}\label{arstist_eqn}
\nu \frac{d \bm{r}_{i}}{dt} = \bm{F}_{i}^{\text{res}}(t),
\end{equation}
where $\nu>0$ denotes the damping constant of the spring force. Equation\autoref{arstist_eqn} is solved using the simple forward Euler method to determine the location of each cell at each timestep, $\Delta t$, see \autoref{tab:chaste parameters} \cite{chapra2012applied}.
\par

 \begin{figure}[h!]
         \centering
         \includegraphics[width=0.9\textwidth]{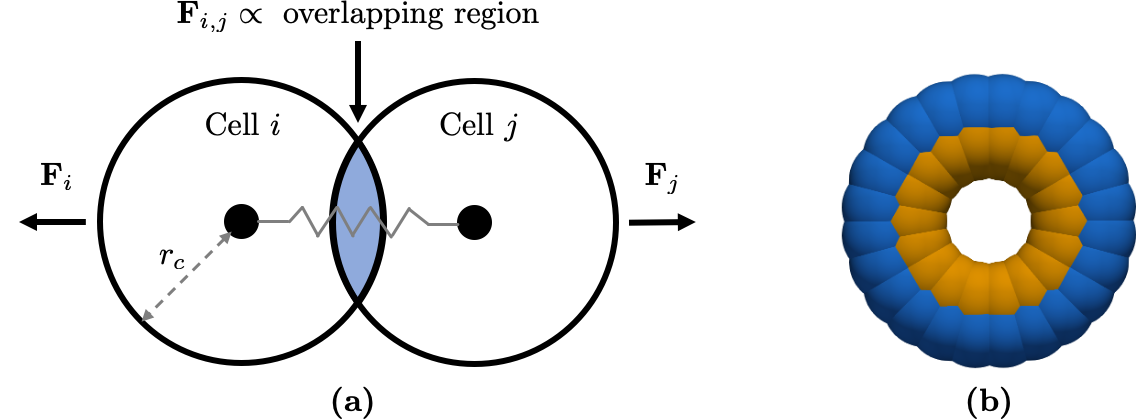}

 			\caption{The cell-based model using the Overlapping Spheres framework. (a) A schematic of the mechanical dynamics that determines the motion of a cell using the Overlapping Spheres framework. The mechanical force acting on each cell is proportional to the region of overlapped space between any two nodes which are the centre of spheres with radius $r_{c}$. The mechanical force between cells $i$ and $j$ can interpreted as a spring force and due to the relevantly low viscosity of the medium, it is assumed that the motion of each cell is governed in an Aristotelian regime, that is, the force is directly proportional to the velocity of the cell. (b) An example of the 2D initial spatial conditions when simulating the bilayer spheroid. The colours of the cells denote cell types, where the blue and orange cells are the basal and luminal cells respectively. The present example has a spheroid radius of 3 cell diameters (CD).}
        \label{fig:ABM_methods}
\end{figure}

Simulations were initialised with a bilayer structure, see \autoref{fig:ABM_methods}b. Basal and luminal cell types were considered to be mechanically identical to isolate the affects of neighbourhood cell-type composition on Delta patterning. Cells were assumed to not proliferate in both 2D and 3D simulations, this was done to control the spatial organisation of cell-types in each layer.\par

The NDM (3.1-3.2) was integrated into each cell in the population and was solved using the explicit Runge-Kutta45 method \cite{chapra2012applied}, which is built into the Chaste software. At every timestep, each cell would sweep through the population to determine the connectivity neighbourhood, which is defined by all nodes within a radius of $\rho_{c}$, as in the fixed geometry simulations. In the simulations presented here, we assume the connectivity radius, $\rho_{c}$, is equal to the mechanical cut-off length, $r_{\text{max}}$. Once a cellular neighbourhood has been determined for each cell, the average Delta is calculated using equation\autoref{eqn:average_delta}, and then updated in the state variables to be used to solve the next timestep of the NDM (3.1-3.2). In all dynamic lattice simulations, we choose $a = 0.01$, $b = 100$, $\mu_{N} = \mu_{D} = 1$, $h = 1$ and $k=2$ as parameter values for the NDM (3.1-3.2), as previously defined \cite{Collier1996}. \par
We measure the existence of laminar patterning of Notch in dynamic domains by taking the difference of the average Notch expression in each cell-type, $\Delta N$. Namely, 
\begin{equation}
\Delta N = \frac{1}{\mathcal{N}_{L}} \sum_{i}^{\mathcal{N}} (1-\delta_{\tau(i),\tau(B)})N_{i}  - \frac{1}{\mathcal{N}_{B}} \sum_{i}^{\mathcal{N}}\delta_{\tau(i),\tau(B)}N_{i},
\end{equation}
where $\mathcal{N}$ is the total number of cells, $\mathcal{N}_{L}$ is the number of luminal cells and $\mathcal{N}_{B}$ is the number of basal cells. The function $\delta_{\tau(i),\tau(B)}$ is cell-type Kronecker delta function,
\begin{equation}
\delta_{\tau(i),\tau(B)} = \left\{
	\begin{array}{ll}
		1  & \mbox{if cell $i$ is a basal cell}, \\
		0 &  \mbox{if cell $i$ is a luminal cell}.
	\end{array}
\right.
\end{equation}
The seeds used to initialise the generation the pseudo-random numbers were fixed for all simulations to compare signal strength parameters on dynamic domains. In addition, $w_{2} = 1$ was fixed for each comparison simulation, see Section 5.4. Parameter values used in all cell-based simulations can be found in \autoref{tab:chaste parameters}.

\begin{table}[H]
\begin{center}
\begin{tabular}{ |c|c|c|c|c| }
 \hline
Parameter & Description & Value & Units & Reference  \\ 
\hline
$t_{\text{tot}}$ & Total simulation time & 100 & h & - \\
$\Delta t$ & Timestep & 0.01 & h & - \\
$\eta_{ij}$ & Spring constant & 25$^{*}$ & NCD$^{-1}$ &  - \\ 
$s_{ij}$ & Spring rest length & 1 & CD & - \\
$r_{\text{max}} $ & Force cut-off length & $3/2^{*}$ & CD & - \\
$k_{c}$ & Decay of force & 5&  \textit{Dim'less} & \cite{Atwell2015} \\
$\xi$  & Random motion perturbation & $0.0025^{*}$& \textit{Dim'less}  & - \\
$\nu$ & Damping constant & 1 &  NhCD$^{-1}$ & \cite{pathmanathan2009computational}\\
\hline
\end{tabular}
\caption{Table of parameters used in each cell-based simulation. The unit of length CD refers to the fixed cell diameter used in simulations. * indicates parameter values tuned for bilayer structure maintenance, the rest of the simulation parameters used in this study were extracted from \cite{Osborne2017}.}\label{tab:chaste parameters}
\end{center}
\end{table}

\subsection{Notch-Delta lateral-inhibition model in static bilayer domains}\label{sec:static_results}

In order to apply both \autoref{prop: bound} and \autoref{coro_stab} to define signal strength parameter regimes for laminar patterning (\autoref{fig:static_results}a), we first rewrite the spatially discrete ODE (3.1-3.2) in the form of a SISO system (3.9-3.10). Let $\bm{x}_{i} = [N_{i},D_{i}]^{T}$ denote the vector of state variables of the system, where $i$ designates cellular identity. Then the input to each cell is the local spatial information received via the $\langle \cdot \rangle$ operator, such that $u_{i} =\langle D_{i} \rangle$. Similarly, the output of each cell is the Delta expression $y_{i} = D_{i}$. To apply \autoref{prop: bound} to our biological model, we need to determine the following: (i) the derivative of the transfer function, $T$, of the SISO system, and (ii) the homogeneous steady state, $\bm{x}^{*}$, of the dynamical system.\par
(i) The derivative of the transfer function can be derived by taking the partial derivative of the SISO system w.r.t. the inputs and state variables, as shown in equation\autoref{eqn:transfer_der}. Thus, for NDM (3.1-3.2),
\begin{equation}\label{eqn:trans_matrices}
\frac{\partial h}{\partial \bm{x}_{i}} = \left[
  \begin{array}{cc}
  0 & 1
   \end{array}
\right], \quad \frac{\partial \bm{f}}{\partial \bm{x}_{i}} =  \left[
  \begin{array}{cc}
  -\mu_{N} & 0\\
  -\frac{bhN_{i}^{h-1}}{ \left( 1+ bN_{i}^{h} \right)^{2}  } &  -\mu_{D}
   \end{array}
\right] \quad \mbox{and} \quad \frac{\partial \bm{f}}{\partial u_{i}} =\left[ \begin{array}{c}
  \frac{aku_{i}^{k-1}}{\left( a + u_{i}^{k} \right)^{2}} \\
  0
   \end{array}
\right].
\end{equation}
Therefore, multiplying the matrices in equation\autoref{eqn:trans_matrices} and making the substitution $N_{i} = f(u_{i})/\mu_{N}$ at steady state, yields the following,
\begin{align}\label{eqn:T_collier}
T^{\prime}(u_{i}) &= -  \left[
  \begin{array}{cc}
  0 & 1
   \end{array}
\right] \left[
  \begin{array}{cc}
  -\mu_{N} & 0\\
  -\frac{bhN_{i}^{h-1}}{ \left( 1+ bN_{i}^{h} \right)^{2}  } &  -\mu_{D}
   \end{array}
\right]^{-1} \left[ \begin{array}{c}
  \frac{aku_{i}^{k-1}}{\left( a + u_{i}^{k} \right)^{2}} \\
  0
   \end{array}
\right], \nonumber \\ &= -\frac{ abkh\mu_{N}^{h} \left( a+u_{i}^{k}  \right)^{h-1} u_{i}^{kh-1} }{   \mu_{D} \left( \mu_{N}^{h}   \left( a + u_{i}^{k} \right)  +bu_{i}^{hk}      \right)^{2} },
\end{align}
which completes the first step in full generality.\par
(ii) We now solve the NDM (3.1-3.2) for the homogeneous steady state. Using the symmetry of homogeneous cell states, this problem is reduced to solving $\bm{f}\left(\bm{x}^{*}, D^{*} \right) = 0$ as $u_{i}^{*} = D_{i}^{*} = D^{*}$ for all $i = 1,...,N$ in the case of a system of identical cells. Solving the system (3.1-3.2) for homogeneous steady states means solving the following polynomial for $D^{*}$,
\begin{equation}\label{eqn:poly_for_D}
b\mu_{D}\left(D^{*}\right) ^{kh+1} + \mu_{N}^{h}\left( \mu_{D}D^{*} - 1 \right) \left( a + \left( D^{*} \right)^{k} \right)^{h} = 0,
\end{equation}
which has analytic solutions for $hk<4$. We choose the parameter values $a = 0.01$, $b = 100$, $\mu_{N} = \mu_{D} = 1$, $h = 1$ and $k=2$ selected previously \cite{Collier1996}, with the only modification of $h$ to simplify calculations for demonstration purposes. To be able to apply \autoref{prop: bound}, we require that $|T^{\prime}(u^{*})|>1$ because $w_{1}>0$. This condition is equivalent to the requirement derived by direct linear analysis as in \cite{Collier1996}, where they show that the existence of homogeneous steady state instability can only occur when
\begin{equation}
f^{\prime}(u^{*})g^{\prime}(N^{*}) < -1.
\end{equation}
As the derivative transfer function $T^{\prime}(u^{*}) = f^{\prime}(u^{*})g^{\prime}(N^{*})/(\mu_{N} \mu_{D})$, using the same parameter values for $\mu_{N}$ and $ \mu_{D}$ as in \cite{Collier1996} yields the equivalent condition. Moreover,  $|T^{\prime}\left(u^{*} \right)|$ is a monotone increasing function with respect to $k$ (see \autoref{fig:mono_k}), hence increasing the nonlinearity of the ODE system relaxes the restrictions on $w_{1}$ for the existence of pattern emergence imposed by \autoref{prop: bound}, therefore emphasising the relationship between the connectivity of the cells and the characteristics of the ODE system.\par

Solving the cubic polynomial\autoref{eqn:poly_for_D} when $k=2$ yields a homogeneous steady state $D^{*} \approx 0.049$, therefore, we have both (i) and (ii). Applying \autoref{prop: bound} to the NDM system (3.1-3.2) using equation\autoref{eqn:T_collier} and $D^{*}$ yields the following bound on signal strength parameters,
\begin{equation}\label{analytical_bound_exist}
w_{1} < \alpha\frac{w_{2}}{R_{\tau}},
\end{equation}
for $\alpha = 0.21$, which defines a strict analytical restriction of the $(w_{1},w_{2})$ parameter space for the emergence of laminar patterning between layers (region below black line in \autoref{fig:static_results}b). 
\par

 \begin{figure}[h!]
         \centering
         \includegraphics[width=0.5\textwidth]{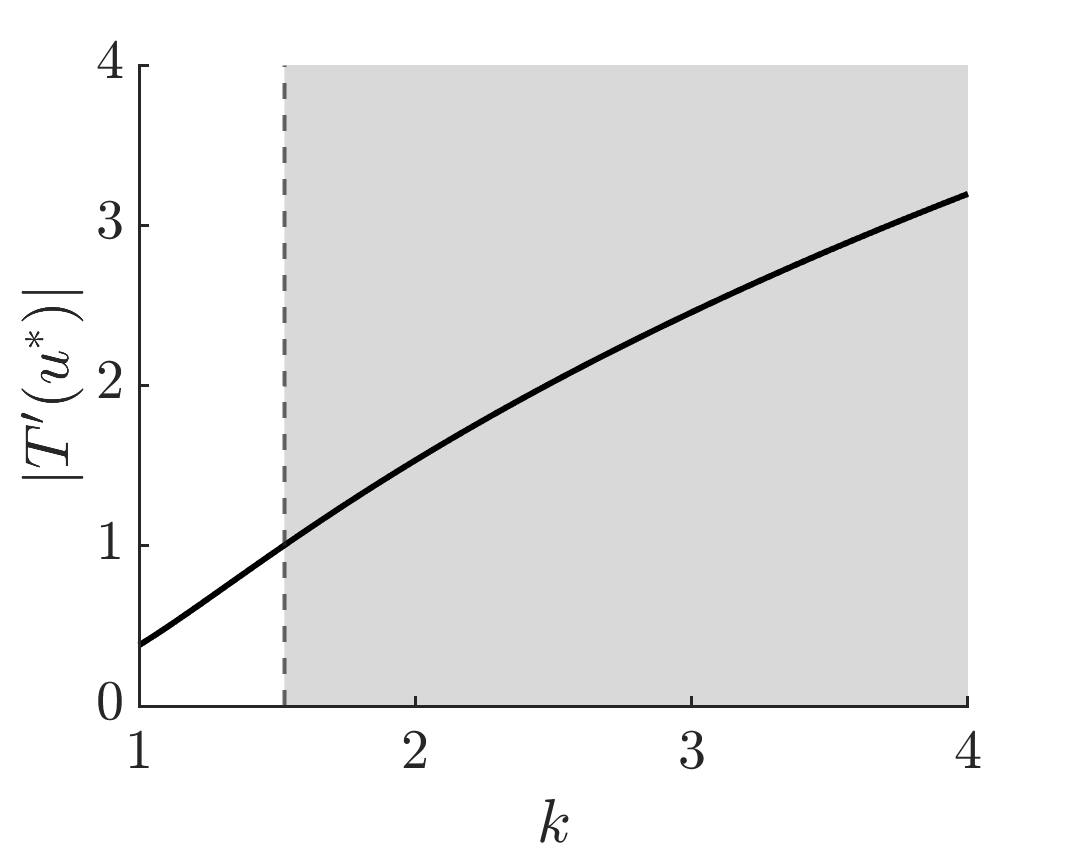}

 			\caption{Monotonicity $|T^{\prime}(u^{*})|$ with respect to $k$. Parameter values were chosen as parameter values $a = 0.01$, $b = 100$, $\mu_{N} = \mu_{D} = 1$ and $h = 1$. For each $k$, the homogeneous steady state was solved using equation\autoref{eqn:poly_for_D}. The shaded region represents the values of $k$ that satisfy the $|T^{\prime}(u^{*})| > 1$, which is required for the instability of the homogeneous steady state in \autoref{prop: bound}, highlighting a lower bound of $k_{\text{min}} = 1.5$.     }
        \label{fig:mono_k}
\end{figure}

As we have found the necessary bound on $w_{1}$ for pattern formation, we now seek to use \autoref{coro_stab} to impose a sufficient bound on $w_{1}$. In order to use \autoref{coro_stab}, we require the $\mathcal{L}_{2}$-gains for each representative cell at steady state. Making use of the monotonicity of the NDM system (3.1-3.2) with respect to the cones $K^{U} = \R_{\geq 0}$, $K^{Y} = -K^{U}$ and $K^{\bm{X}} = \{ \bm{x} \in  \R^{2} | x_{1} \geq 0, x_{2} \leq 0 \} $ \cite{Arcak2013}, we are able to use the steady state relation\autoref{eqn:l2_t}. To determine the $\mathcal{L}_{2}$-gains, we solve for the heterogeneous steady states $\bm{x}_{B}$ and $\bm{x}_{L}$, with associated input steady states
\begin{equation}
u_{B} = \frac{n_{1}w_{1}D_{B} + n_{2}w_{2}D_{L} }{N_{w}} \quad \mbox{and} \quad u_{L} = \frac{n_{1}w_{1}D_{L} + n_{2}w_{2}D_{B} }{N_{w}},
\end{equation}
then using equation\autoref{eqn:T_collier}, $q_{B} = |T^{\prime} \left( u_{B}\right)| $ and $q_{L} =| T^{\prime} \left( u_{L}\right)|  $. For each static geometry outlined in \autoref{tab:geometry sum}, a parameter sweep of the signal strength parameter space $\left( w_{1}, w_{2}\right)$ was conducted to highlight regions that satisfy the inequality\autoref{eqn_full_stab_bound} where the heterogeneous steady states $\bm{x}_{B}$ and $\bm{x}_{L}$ were numerically solved using the \texttt{fsolve} function in Matlab 2019b. The resulting stability regions in the $(w_{1},w_{2})$-space (red shaded regions in \autoref{fig:static_results}b) has the same linear form as the analytical existence bound\autoref{analytical_bound_exist}. Therefore, assuming the same form of relationship between $w_{1}$, $w_{2}$ and $R_{\tau}$, a ubiquitous gradient parameter $\beta$ was extracted from each static lattice parameter sweep. That is, to ensure the local asymptotic stability of the laminar bilayer patterns (\autoref{fig:static_results}a) in both 2D and 3D, 
\begin{equation}\label{example:stab_bound}
w_{1}  < \beta \frac{w_{2}}{R_{\tau}}
\end{equation}
must be satisfied, for $\beta = 0.04$. We have provided an improvement on equation\autoref{analytical_bound_exist} from necessary to sufficient for laminar pattern formation using the NDM system (3.1-3.2), nevertheless, this defines a highly restrictive parameter bound on $w_{1}$.\par
Finally, using static lattice simulations for each of 2D and 3D geometries described in \autoref{tab:geometry sum} (see \autoref{sec:static_methods}) we conducted the same parameter sweep over the $(w_{1},w_{2})$-space to verify the necessary bound of inequality\autoref{analytical_bound_exist} and the sufficient bound of inequality\autoref{example:stab_bound}. The parameter regions that exhibited the layered patterning using a pattern tolerance of $\delta = 0.1$ were consistent with the analytical inequalities\autoref{analytical_bound_exist} and\autoref{example:stab_bound}, as demonstrated in \autoref{fig:static_results}b by the blue shaded regions. Furthermore, the regions defining the observed patterns from numerical simulations had the same linear upper bound for $w_{1}$ as a function of $w_{2}$ for all 2D and 3D geometries. Therefore, as conducted for the stability inequality\autoref{example:stab_bound}, we extracted a ubiquitous gradient parameter $\gamma$, such that laminar patterning can be observed in a bilayer of cells if $w_{1}$ satisfies,
\begin{equation}
w_{1}< \gamma \frac{w_{2}}{R_{\tau}}
\end{equation}
where $\gamma = 0.11$. Note that due to the symmetry of the system in order to achieve the laminar patterning in the correct direction, the system required a small perturbation using initial conditions. Moreover, as the pattern tolerance $\delta \rightarrow 0$ then $\gamma \rightarrow \alpha$, due to the contrast between the layers becoming much weaker, see \autoref{fig:2d_abm}b. Thus the arbitrary choice of $\delta$ defines what is considered as acceptable patterning, though we note that the necessary bound provided by \autoref{prop: bound} is always satisfied.\par 

As the observed pattern regions lie within the existence bound regions and the sufficient stability bound regions are a subset of the observed pattern regions in $(w_{1},w_{2})$-space (\autoref{fig:static_results}b), we numerically verify the analytical conditions imposed on the signal strength parameters $w_{1}$ and $w_{2}$ by \autoref{prop: bound} and \autoref{coro_stab} using the NDM system (3.1-3.2). In each case, for existence, stability and numerical observation, there exists a consistent form for the upper bound of the cell-type dependent signal strength parameter $w_{1}$, which relates cellular connectivity to signal strength polarisation, via $R_{\tau}$, independent of lattice dimension. Thus, static domains may no longer be required to impose conditions for laminar pattern formation. Knowing only the cell-type composition of the neighbourhood for each cell may be sufficient for an adaptive signal strength to maintain pattern formation, and therefore would allow for juxtacrine dependent pattern formation to be studied in dynamic cellular domains, which will be discussed further in \autoref{sec:abm_results}.

 \begin{figure}[H]
  \vspace{0cm}
         \centering
         \includegraphics[width=0.75\textwidth]{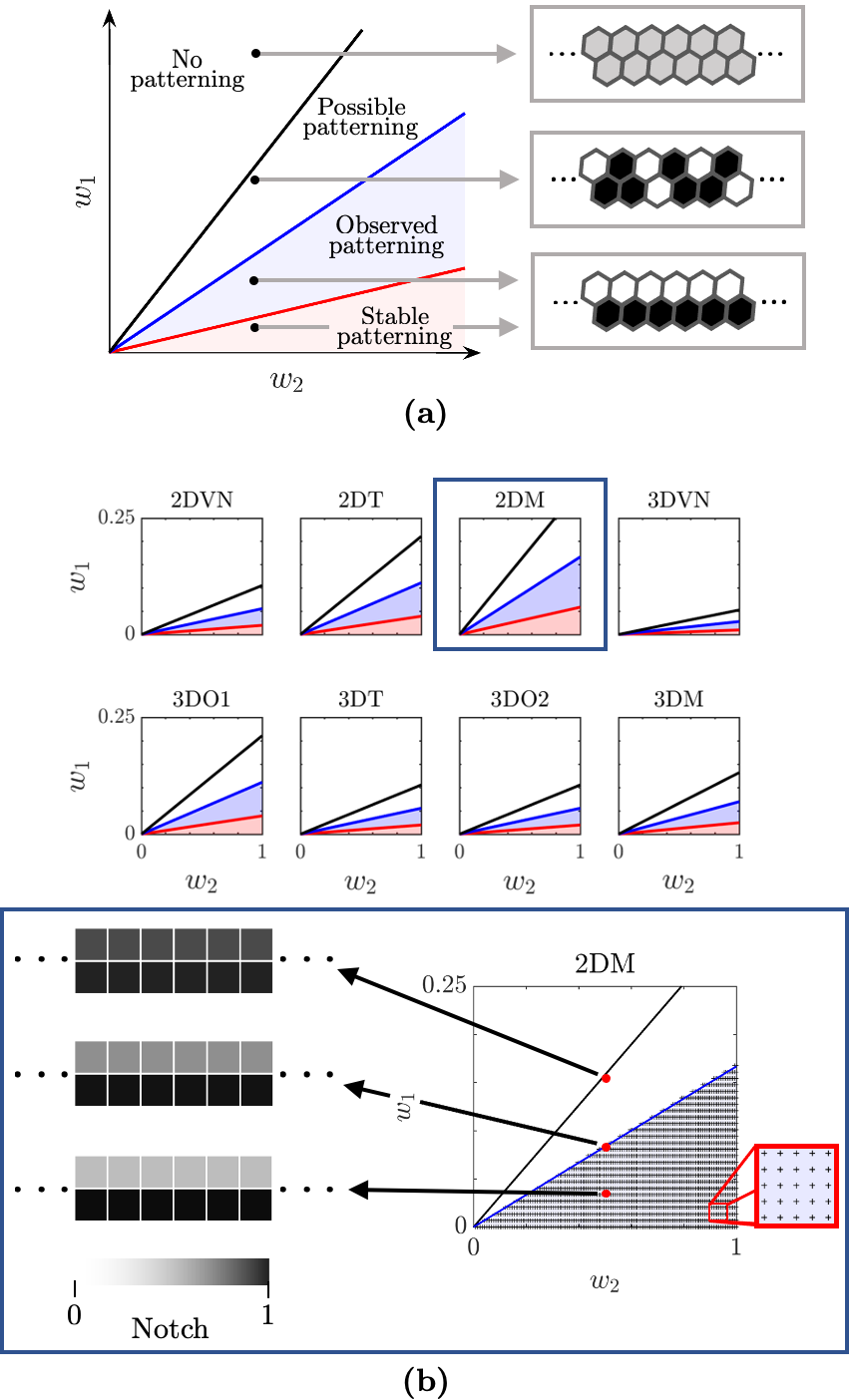}

 			\caption{Signal strength polarisation conditions in static bilayer geometries. (a) A representative diagram of the $(w_{1},w_{2})$ parameter space that yields conditions for bilayer laminar patterning. The region above the black line corresponds to stability of the  steady state, where the black line is the upper bound of $w_{1}$ provided by \autoref{prop: bound}. (Continued on the following page.) }
        \label{fig:static_results}
\end{figure}

\begin{figure}[t]
  \contcaption{(Continued.) The blue line is the upper bound of $w_{1}$ determined from numerical simulations of the ODE system and the red line is the upper bound for analytical stability of the heterogeneous steady states provided by \autoref{coro_stab}.  Representative patterns are embedded into the 2D Triangulated lattice.  (b) $(w_{1},w_{2})$ parameter space highlighting laminar pattern regions shown in (a) for each static geometry outlined in \autoref{tab:geometry sum} with example simulation results using the 2DM lattice. The magnified region of $(w_{1},w_{2})$-space demonstrates the high density of parameter values with the capacity to produce laminar patterning, denoted by +,  which defines the blue observed regions in all static lattice simulations. Red points represent the parameter values used in the example simulations.}
\end{figure}

\subsection{Notch-Delta lateral-inhibition model in dynamic bilayer domains}\label{sec:abm_results}
In \autoref{sec:static_results}, analytical and numerical bounds on $w_{1}$ were derived for the existence and stability of Notch-Delta polarisation in static bilayer cell domains. Motivated by the consistency of the laminar patterning of Notch in the developing mammary gland (see \autoref{fig:mammary_organoid_intro}), we seek to test if the static bounds derived on $w_{1}$ are able produce the same pattern formation in a dynamic domains. In addition, we investigate the efficacy of using a fixed and adaptive upper bound of $w_{1}$ to ensure laminar pattern stability in bilayers.\par

When transitioning to dynamic domains, we cannot always satisfy the \textit{equitable} property of the cell-type partitions $P_{B}$ and $P_{L}$ in the connected graph. Consequently, the analytical conditions derived in \autoref{sec:theoretical_results} cannot be applied at each timestep of the simulations, instead, we use the static domain inequalities (5.14-5.15) to gain intuition for polarisation conditions in dynamic geometries. In particular, we focus on how a cell responds to the microenvironment via two cell-type signal strength mechanisms: (i) globally fixed values of $w_{1}$ and $w_{2}$ and (ii) a locally adaptive $w_{1}$ for a globally fixed $w_{2}$. By investigating these two types of signal strength mechanisms in the dynamic cellular domains, we are able to measure the influence of varying cellular connections on pattern stability as the system evolves.\par

The \textit{fixed} mechanism for $w_{1}$ (case (i)) is used to represent a high inertia of cellular adaptability to the local environment of the cell, that is, that transmission strengths are defined at birth. Here, the $w_{1}$ is set to agree with the inequalities (5.14) and (5.15), designated as ``Fixed $\beta$'' (F$\beta$) and ``Fixed $\gamma$'' (F$\gamma$), respectively. Using the \textit{fixed} mechanism in dynamics simulations, the $R_{\tau}$ value is defined by the initial connectivity of the geometry and is constant throughout the simulation.\par

The \textit{adaptive} mechanism for $w_{1}$ (case (ii)) is used to represent a low inertia of cellular response to the microenviroment. That is, for each cell, $w_{1}$ is updated at each timestep to satisfy the observed static inequality (5.15) by determining $R_{\tau,i}$, i.e. the cell-type composition of the neighbourhood for each cell $i$ (see \autoref{sec:abm_methods}). We denote this signal strength mechanism ``Adaptive $\gamma$'' (A$\gamma$). \par

Simulating each signal strength mechanism, \textit{fixed} (i) and \textit{adaptive} (ii) for 100 hours demonstrates that conditions defining laminar pattern regions in static geometrics, inequalities (5.14-5.15), allow for the emergence of laminar patterns in dynamic cell geometries up to small spatial perturbations, \autoref{fig:2d_abm}a. That is, each mechanism initially produced concentric contrasting layers of Notch expression, however, as the bilayer layer geometry became deformed due to the random perturbations of each cell, the definition between layers was lost by 100 hours (Figures 12a-12c). Thus, information about cell-type signal polarisation is preserved when partitions of the connected graph are no longer \textit{equitable}, however, the retained information is insufficient for long-term stability of the Notch states.\par

In terms of pattern intensity and retention, the simulation using Fixed $\beta$ performed the best, though due to the high contrast between layers, the variance in Notch expression quickly becomes very large once consistent patterning is lost, Figures 12b-12c. The region of $(w_{1},w_{2})$-space, defined by inequality (5.14), which is sufficient for stability of heterogeneous states between layers of static bilayers is highly restrictive, such that $w_{1} \approx 0$ for all 2D and 3D geometries. In context, the simulation using Fixed $\beta$ accounts for the situation of no Delta signalling between cells of the same type. \par 

The Fixed $\gamma$ signalling mechanism produced the least contrast of Notch expression between layers initially, and was quick to lose the consistency of expression, therefore performing the worst out of the signal strength mechanisms (Figures 12b-12c). However, the Adaptive $\gamma$ signal strength mechanism yielded the greatest pattern retention over the total time, highlighted by the lowest variance from $\Delta N$ values Figures 12b-12c. As Adaptive $\gamma$ allows for $w_{1} \approx \gamma$ (see \autoref{fig:2d_abm}d), this ability of the cell to update signal strength dependent on the local cell-type composition, enables cells to still signal to cells of the same type whilst maintaining the concentric patterning. This highlights that if homophilic signalling is observed, then cells may be actively adapting to the microenvironment to stabilise stratification. \par

Furthermore, using the Adaptive $\gamma$ signal strength mechanism revealed that there is stricter polarisation conditions in the basal cells than luminal cells while laminar patterning is maintained, for $t < 50$h (Figures 12a-12c). That is, due to the geometry of the cellular domain, basal cells are less connected to the luminal layer than the luminal cells are to the basal. Hence, by the inverse relationship between the cell-type connectivity and lateral-inhibition model (\autoref{prop: bound}), the restricted cellular signalling imposed on the basal cells, may induce laminar pattern formation within the luminal cells, whilst allowing for greater luminal-luminal cell communication (\autoref{fig:2d_abm}d). Moreover, at $t \approx 50$h, a basal cell was disconnected from the luminal layer, producing a transient irregularity for $w_{1}$ values (bounded above inequality (5.12)), and therefore initiating the deterioration of laminar pattering \autoref{fig:2d_abm}d. 

 \begin{figure}[H]
         \centering
         \includegraphics[width=1\textwidth]{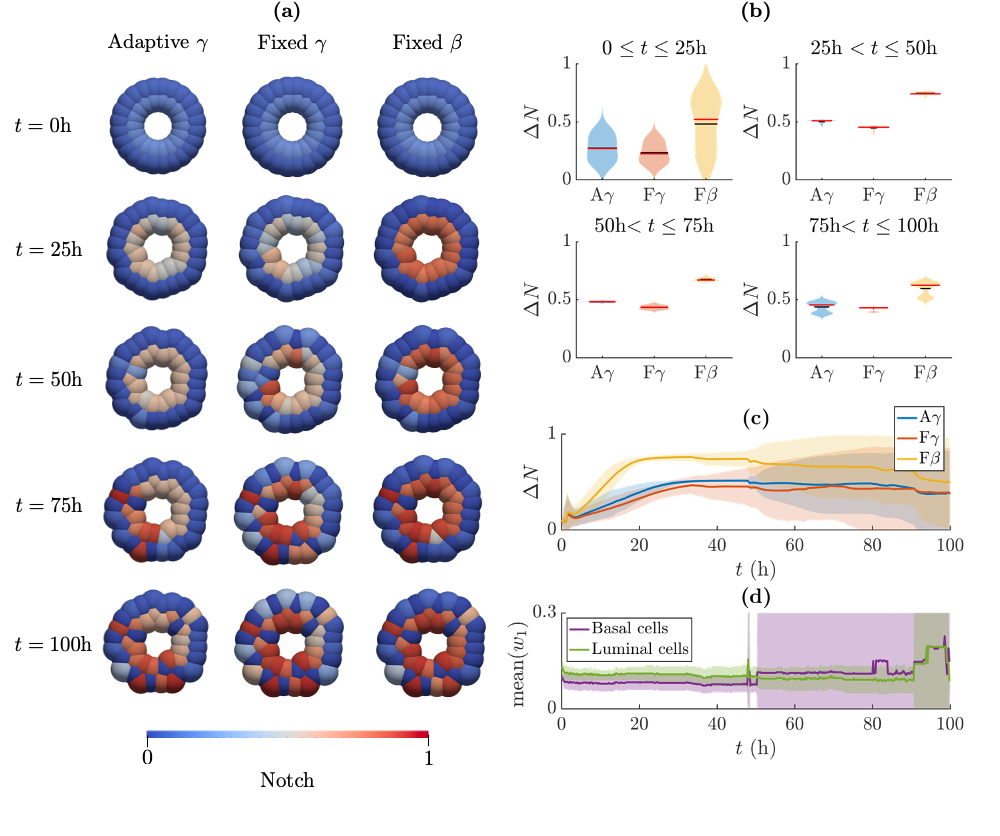}

 			\caption{2D dynamic cellular domain simulation results of fixed and adaptive signal strength polarisation with each simulation using a globally fixed value of $w_{2} = 1$. (a) Cell-based simulations of a cross-section of a bilayer spheroid. The simulations were run for 100 hours, initialised with cell-type stratification (see \autoref{fig:ABM_methods}b), and ODE initial conditions $\bm{x}_{B}(0) = [0.1,0.2]^{T}$ and $\bm{x}_{L}(0) = [0.2,0.1]^{T}$. The colour of each represents the intracellular level of Notch. (b) Violin plots summarising the results of (a). Shaded regions denote the probability density of the $\Delta N$ values. The black and red lines are the means  and medians of the $\Delta N$ values, respectively. (c) A plot of the $\Delta N$ value for each signal polarity mechanism over time. Shaded regions represent standard deviations from the mean Notch expression of each cell-type. (d) An additional output plot for the adaptive signalling mechanism demonstrating the disparity of $w_{1}$ values for basal and luminal cells over time. Shaded regions represent standard deviations from the mean $w_{1}$ of each cell-type.}
        \label{fig:2d_abm}
\end{figure}

Cell-based simulations investigating the efficacy of the static geometry polarisation conditions in dynamic domains were initially conducted on 2D cross-sections of bilayer spheroids. Analysis conducted on static geometries suggested that the signal strength conditions on $w_{1}$ are independent of physical dimension. We show in Figures 13a-13b that simulations of 3D spheroids are in agreement with 2D cross-sections, namely, that both the \textit{fixed} and \textit{adaptive} signalling mechanism are capable for generating laminar patterning but are unable to retain the definition of the layers for long periods. Due to the increase of cells in 3D simulations, there is a large increase in the amount of spatial perturbations of the cellular structure due to the random motion applied to each cell. As a result, the time at which consistent laminar patterning is broken occurs much earlier, at $t \approx 20$h (\autoref{fig:3d_amb}b). It should also be noted that initial spatial conditions are not identical as in the 2D simulations as currently there exists no solution to map equidistant points that cover the surface a sphere. We instead use the Fibonacci spiral method as an approximate solution, though, this produces clustering of cells the poles of the sphere \cite{gonzalez2010measurement}, and so introduces initial artefacts to cellular connectivity in the simulations.

 \begin{figure}[H]
         \centering
         \includegraphics[width=0.8\textwidth]{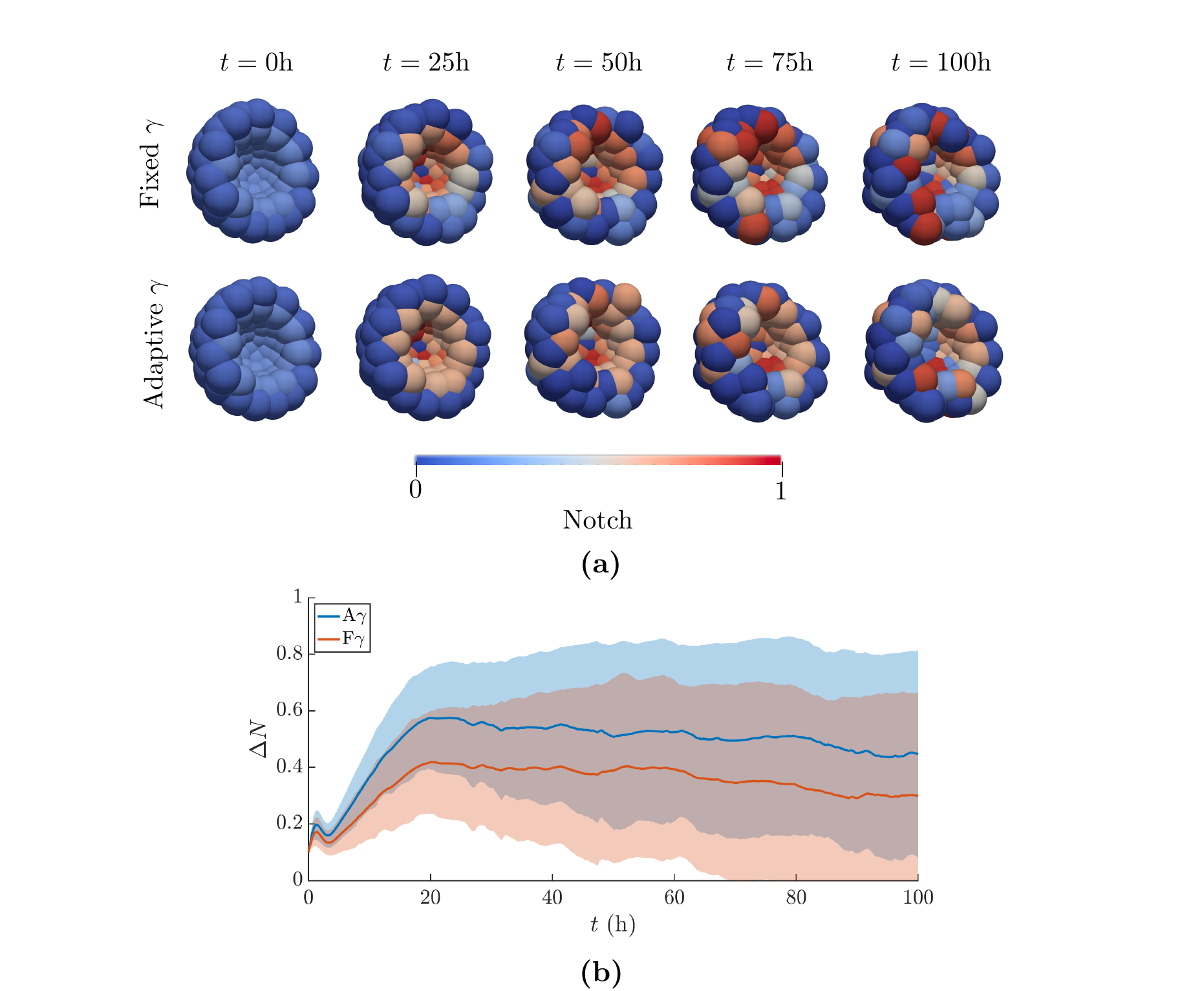}

 			\caption{3D cell-based simulation of a bilayer spheroid representing a developing mammary organoid using the \textit{adaptive} and \textit{fixed} signal strength mechanisms. (a) The simulations were run for 100 hours, initialised with cell-type stratification (see \autoref{fig:ABM_methods}b), and ODE initial conditions $\bm{x}_{B}(0) = [0.1,0.2]^{T}$ and $\bm{x}_{L}(0) = [0.2,0.1]^{T}$. The colour of each represents the intracellular level of Notch. Only half of the spheroids are visualised to show the dynamics of the internal luminal cells. (b) A plot of the $\Delta N$ value for both Fixed $\gamma$ and Adaptive $\gamma$ signal polarity mechanisms over time for the 3D simulations. Shaded regions represent standard deviations from the mean Notch expression of each cell-type.}
        \label{fig:3d_amb}
\end{figure}

\section{Discussion}

We have developed a framework for investigating cell-type juxtacrine signal strength polarisation conditions for emergence and stability of laminar patterning in symmetric bilayer structures using lateral-inhibition ODE systems. Leveraging previous results of graph partitioning on monolayers, we show how the geometry of the cellular domain has a large impact on the systems capacity to produce heterogeneity. Moreover, using this framework, we replace the algebraically demanding process of linear analysis of large multicellular systems with an exploitation of the spectral properties of the connected graph, therefore addressing the complexity issue discussed in previous juxtacrine pattern analysis studies \cite{Wearing2000}.\par  

In \autoref{sec:theoretical_results}, we provide necessary and sufficient conditions for the existence of laminar patterns in a bilayer of cells. Both existence and stability inequalities\autoref{eqn: existence_bound} and\autoref{eqn_full_stab_bound} highlight that upon increasing connectivity with opposing cell-types allows for larger existence and stability regions in $(w_{1},w_{2})$-space. In context of a bilayer of cells, as global concavity of the structure increases, luminal cells have a greater probability to connect with more basal cells. Therefore, relaxing the existence and stability conditions imposed by \autoref{prop: bound} and \autoref{coro_stab} by decreasing $R_{\tau}$. However, this would violate the symmetry between partitions required to apply both \autoref{prop: bound} and \autoref{coro_stab}, hence we propose that by investigating asymmetric connections between layers of cells may allow for a relationship between global curvature of the cellular structure and pattern stability. Furthermore, as the signal polarisation constraints, inequalities (4.1) and (4.5), presented in this study are independent of the lateral-inhibition model, \autoref{prop: bound} emphasises the influence of connectivity on pattern formation by requiring cell-type signal strength heterogeneity. Therefore, in cell-type stratified bilayer geometries, cell-type dependent signal strength polarisation is required for the existence of heterogeneous steady states, which is independent of the feedback mechanism. This indicates the critical role of cellular connectivity in juxtacrine systems and therefore the geometry of the cellular domain should be carefully considered in patterning formation studies. 

By studying a family of 2D and 3D static cellular domains of varying connectivity, we are able to gain insight into the emergence and stability of centric layer pattern formations in dynamic domains. Namely, by employing the bounds on $w_{1}$ derived from the static simulations, we were able to generate the laminar patterns in 2D and 3D bilayer spheroids when imposing random spatial perturbations on each cell. However, these patterns became unstable as geometry deformation increased, producing disorganised layers of Notch expression even when using an adaptive polarisation mechanism (\autoref{fig:2d_abm} and \autoref{fig:3d_amb}). Therefore, the information obtained from static domains is insufficient to fully characterise the behaviour of the lateral-inhibition model in a developing mammary organoid. Although, in this study we assume that the laminar pattern formation is driven purely by signal strength polarisation between the layers, thus neglecting the affect of the external environment on the biological system. That is, we neglect the influence of stroma or extracellular matrix and the importance of the lumen to the luminal cells in supporting high contrast of Notch expression \textit{in vivo} and \textit{in vitro}, respectively \cite{Lloyd-Lewis2019}. Thus, applying supplementary boundary conditions in dynamic domains in addition to signal polarisation may achieve laminar patterning, invariant to morphological perturbations. 

Applying the analytical polarisation conditions of \autoref{prop: bound} and \autoref{coro_stab} to the context of a mammary organoid using the Collier et al. (1996) NDM revealed that if patterns are to be experimentally observed then we required almost no juxtacrine communication between cells within the same layer (Figures 2b-2d). A plausible process to address the polarisation of Notch expression could involve cadherins \cite{Knudsen2005}, which are transmembrane proteins that mediate cell-cell adhesions. Differential expression of cadherins (E-cadherins are associated with luminal cells and P-cadherins are  with basal cells) are suggested to promote self-organisation to form bilayer structures in the mammary gland via cellular affinity to homophilic interactions \cite{Knudsen2005}. There is growing evidence for an inverse relationship between Notch and E-cadherins in biological systems, including mammary epithelia \cite{Pereira2006,Leong2007,chen2010hypoxia}. In addition, it has been verified that E-cadherins located between luminal cells promote lumen formation during mammary organoid development \cite{Shamir2014}. Therefore, we propose that there exists a cadherin adhesion dependent Notch inhibition mechanism that promotes the Notch signalling between layers of cells (\autoref{fig:cadherins}). \par

 \begin{figure}[H]
         \centering
         \includegraphics[width=0.6\textwidth]{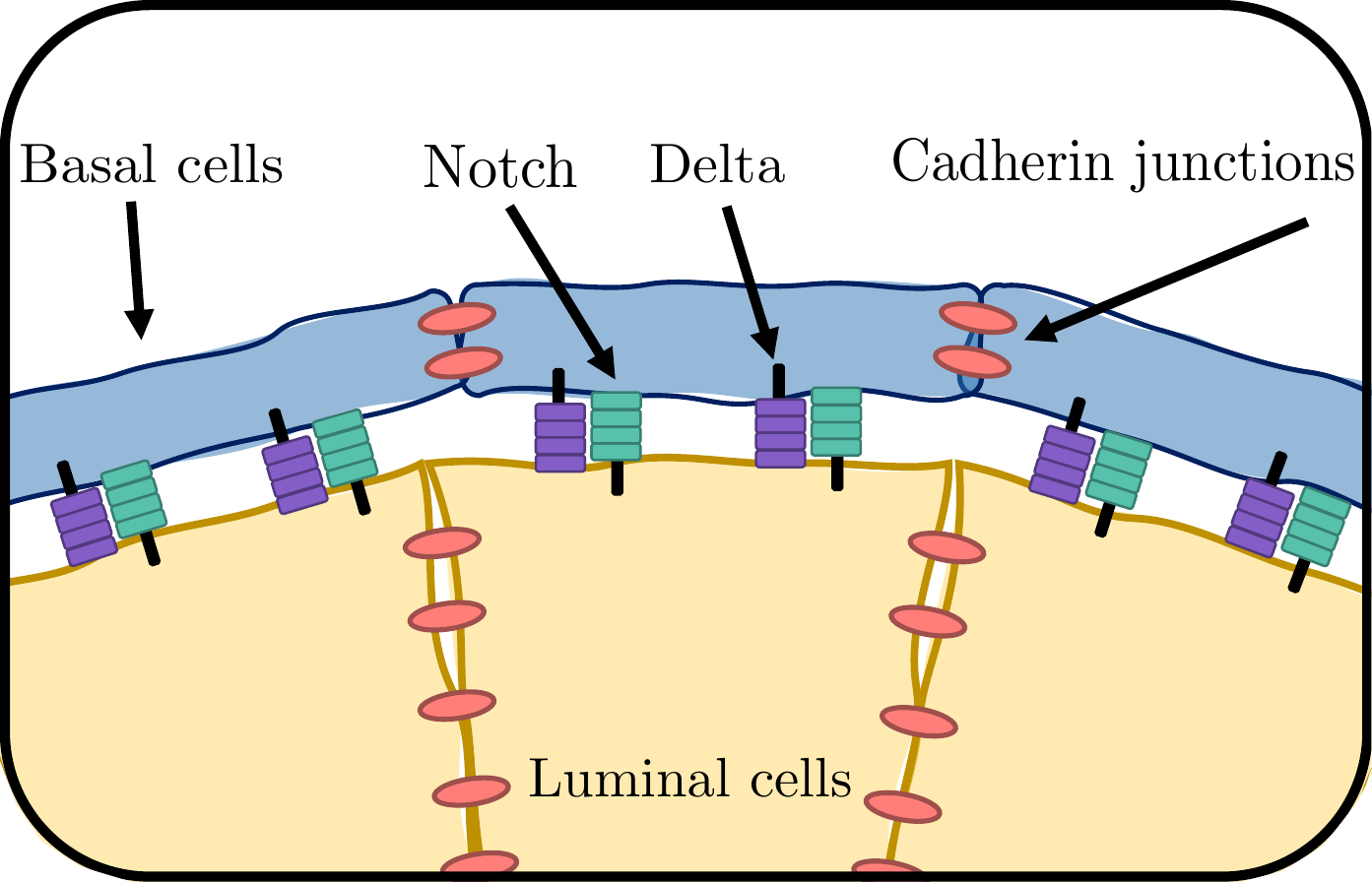}

 			\caption{Proposed spatial distribution of Notch, Delta and Cadherin junctions within a developing mammary organoid. Due to the adhesion required to maintain the bilayer structure with a hollow lumen, tight junctions form, inhibiting the function of the membrane bound Notch receptors and Delta ligands between cells in the same layers.}
        \label{fig:cadherins}
\end{figure}

\section*{Acknowledgements}
JWM is supported by Knowledge Economy Skills Scholarships (KESS2), a pan-Wales higher-level skills initiative led by Bangor University on behalf of the Higher Education sector in Wales. It is part-funded by the Welsh Government’s European Social Fund (ESF).

\section*{Supporting information}
Source code for both dynamic and static lattice simulations can be found at \url{http://bit.ly/Bilayer_sims_repo}.

\bibliographystyle{unsrtnat}
 \bibliography{IMA_bI_ND_Graph_arxiv}

\begin{thebibliography}{47}
\providecommand{\natexlab}[1]{#1}
\providecommand{\url}[1]{\texttt{#1}}
\expandafter\ifx\csname urlstyle\endcsname\relax
  \providecommand{\doi}[1]{doi: #1}\else
  \providecommand{\doi}{doi: \begingroup \urlstyle{rm}\Url}\fi

\bibitem[Gilbert(2016)]{GilbertScottF.2016Db}
Scott~F. Gilbert.
\newblock \emph{Developmental biology}, chapter~1.
\newblock Sinauer Associates Inc, eleventh edition. edition, 2016.

\bibitem[Cagan(2009)]{cagan2009principles}
Ross Cagan.
\newblock Principles of drosophila eye differentiation.
\newblock \emph{Current topics in developmental biology}, 89:\penalty0
  115--135, 2009.
\newblock \doi{10.1016/S0070-2153(09)89005-4}.

\bibitem[Togashi et~al.(2011)Togashi, Kominami, Waseda, Komura, Miyoshi,
  Takeichi, and Takai]{togashi2011nectins}
Hideru Togashi, Kanoko Kominami, Masazumi Waseda, Hitomi Komura, Jun Miyoshi,
  Masatoshi Takeichi, and Yoshimi Takai.
\newblock Nectins establish a checkerboard-like cellular pattern in the
  auditory epithelium.
\newblock \emph{Science}, 333:\penalty0 1144--1147, 2011.
\newblock \doi{10.1126/science.1208467}.

\bibitem[Kim et~al.(2015)Kim, Huang, Critser, Yang, Chan, Wang, Voytik-Harbin,
  Bernstein, and Yoder]{kim2015notch}
Hyojin Kim, Lan Huang, Paul~J. Critser, Zhenyun Yang, Rebecca~J. Chan, Lin
  Wang, Sherry~L Voytik-Harbin, Irwin~D. Bernstein, and Mervin~C. Yoder.
\newblock Notch ligand delta-like 1 promotes in vivo vasculogenesis in human
  cord blood--derived endothelial colony forming cells.
\newblock \emph{Cytotherapy}, 17:\penalty0 579--592, 2015.
\newblock \doi{10.1016/j.jcyt.2014.12.003}.

\bibitem[Turing(1952)]{turing1952chemical}
Alan~M. Turing.
\newblock The chemical basis of morphogenesis.
\newblock \emph{Philosophical Transactions of the Royal Society of London.
  Series B, Biological Sciences}, 237\penalty0 (641):\penalty0 37--72, 1952.
\newblock \doi{10.1098/rstb.1952.0012}.

\bibitem[Hadjivasiliou et~al.(2016)Hadjivasiliou, Hunter, and
  Baum]{Hadjivasiliou2016}
Zena Hadjivasiliou, Ginger~L. Hunter, and Buzz Baum.
\newblock A new mechanism for spatial pattern formation via lateral and
  protrusionmediated lateral signalling.
\newblock \emph{Journal of the Royal Society Interface}, 13, 2016.
\newblock \doi{10.1098/rsif.2016.0484}.

\bibitem[Manz and Groves(2010)]{manz2010spatial}
Boryana~N. Manz and Jay~T. Groves.
\newblock Spatial organization and signal transduction at intercellular
  junctions.
\newblock \emph{Nature Reviews Molecular Cell Biology}, 11:\penalty0 342--352,
  2010.
\newblock \doi{10.1038/nrm2883}.

\bibitem[Bray(2006)]{Bray2006}
Sarah~J. Bray.
\newblock Notch signalling: A simple pathway becomes complex.
\newblock \emph{Nature Reviews Molecular Cell Biology}, 7:\penalty0 678--689,
  2006.
\newblock \doi{10.1038/nrm2009}.

\bibitem[Collier et~al.(1996)Collier, Monk, Maini, and Lewis]{Collier1996}
Joanne~R. Collier, Nicholas~A.M. Monk, Philip~K. Maini, and Julian~H. Lewis.
\newblock Pattern formation by lateral inhibition with feedback: A mathematical
  model of delta-notch intercellular signalling.
\newblock \emph{Journal of Theoretical Biology}, 1996.
\newblock \doi{10.1006/jtbi.1996.0233}.

\bibitem[Wearing and Sherratt(2001)]{Wearing2001}
Helen~J. Wearing and Jonathan~A. Sherratt.
\newblock Nonlinear analysis of juxtacrine patterns.
\newblock \emph{SIAM Journal on Applied Mathematics}, 62:\penalty0 283--309,
  2001.
\newblock \doi{10.1137/S003613990037220X}.

\bibitem[Webb and Owen(2004{\natexlab{a}})]{Webb2004}
Steven~D. Webb and Markus~R. Owen.
\newblock Oscillations and patterns in spatially discrete models for
  developmental intercellular signalling.
\newblock \emph{Journal of Mathematical Biology}, 48:\penalty0 444--476,
  2004{\natexlab{a}}.
\newblock \doi{10.1007/s00285-003-0247-1}.

\bibitem[Webb and Owen(2004{\natexlab{b}})]{Webb2004_osc}
Steven~D. Webb and Markus~R. Owen.
\newblock Intra-membrane ligand diffusion and cell shape modulate juxtacrine
  patterning.
\newblock \emph{Journal of Theoretical Biology}, 230:\penalty0 99--117,
  2004{\natexlab{b}}.
\newblock \doi{10.1016/j.jtbi.2004.04.024}.

\bibitem[Formosa-Jordan and Iba{\~{n}}es(2014)]{Formosa-Jordan2014}
Pau Formosa-Jordan and Marta Iba{\~{n}}es.
\newblock Competition in notch signaling with cis enriches cell fate decisions.
\newblock \emph{PLoS ONE}, 9, 2014.
\newblock \doi{10.1371/journal.pone.0095744}.

\bibitem[Lilja et~al.(2018)Lilja, Rodilla, Huyghe, Hannezo, Landragin, Renaud,
  Leroy, Rulands, Simons, and Fre]{Lilja2018}
Anna~M. Lilja, Veronica Rodilla, Mathilde Huyghe, Edouard Hannezo, Camille
  Landragin, Olivier Renaud, Olivier Leroy, Steffen Rulands, Benjamin~D.
  Simons, and Silvia Fre.
\newblock Clonal analysis of notch1-expressing cells reveals the existence of
  unipotent stem cells that retain long-term plasticity in the embryonic
  mammary gland.
\newblock \emph{Nature Cell Biology}, 20:\penalty0 677--687, 2018.
\newblock \doi{10.1038/s41556-018-0108-1}.

\bibitem[Hamada et~al.(2014)Hamada, Watanabe, Lau, Nishida, Hasegawa, Parichy,
  and Kondo]{hamada2014involvement}
Hiroki Hamada, Masakatsu Watanabe, Hiu~Eunice Lau, Tomoki Nishida, Toshiaki
  Hasegawa, David~M Parichy, and Shigeru Kondo.
\newblock Involvement of delta/notch signaling in zebrafish adult pigment
  stripe patterning.
\newblock \emph{Development}, 141:\penalty0 318--324, 2014.
\newblock \doi{10.1242/dev.099804}.

\bibitem[Arcak(2013)]{Arcak2013}
Murat Arcak.
\newblock Pattern formation by lateral inhibition in large-scale networks of
  cells.
\newblock \emph{IEEE Transactions on Automatic Control}, 58:\penalty0
  1250--1262, 2013.
\newblock \doi{10.1109/TAC.2012.2231571}.

\bibitem[{Rufino Ferreira} and Arcak(2013)]{RufinoFerreira2013}
Ana~S. {Rufino Ferreira} and Murat Arcak.
\newblock Graph partitioning approach to predicting patterns in lateral
  inhibition systems.
\newblock \emph{SIAM Journal on Applied Dynamical Systems}, 12:\penalty0
  2012--2031, 2013.
\newblock \doi{10.1137/130910142}.

\bibitem[J{\o}rgensen et~al.(2018)J{\o}rgensen, de~Lichtenberg, Collin, Klinck,
  Ekberg, Engelstoft, Lickert, and Serup]{Jorgensen2018}
Mette~C. J{\o}rgensen, Kristian~H. de~Lichtenberg, Caitlin~A. Collin, Rasmus
  Klinck, Jeppe~H. Ekberg, Maja~S. Engelstoft, Heiko Lickert, and Palle Serup.
\newblock Neurog3-dependent pancreas dysgenesis causes ectopic pancreas in hes1
  mutant mice.
\newblock \emph{Development (Cambridge)}, 145:\penalty0 1--11, 2018.
\newblock \doi{10.1242/dev.163568}.

\bibitem[Fre et~al.(2005)Fre, Huyghe, Mourikis, Robine, Louvard, and
  Artavanis-Tsakonas]{Fre2005}
Silvia Fre, Mathilde Huyghe, Philippos Mourikis, Sylvie Robine, Daniel Louvard,
  and Spyros Artavanis-Tsakonas.
\newblock Notch signals control the fate of immature progenitor cells in the
  intestine.
\newblock \emph{Nature}, 435\penalty0 (7044):\penalty0 964--968, 2005.
\newblock \doi{10.1038/nature03589}.

\bibitem[Lafkas et~al.(2015)Lafkas, Shelton, Chiu, {De Leon Boenig}, Chen,
  Stawicki, Siltanen, Reichelt, Zhou, Wu, Eastham-Anderson, Moore, Roose-Girma,
  Chinn, Hang, Warming, Egen, Lee, Austin, Wu, Payandeh, Lowe, and
  Siebel]{Lafkas2015}
Daniel Lafkas, Amy Shelton, Cecilia Chiu, Gladys {De Leon Boenig}, Yongmei
  Chen, Scott~S. Stawicki, Christian Siltanen, Mike Reichelt, Meijuan Zhou,
  Xiumin Wu, Jeffrey Eastham-Anderson, Heather Moore, Meron Roose-Girma, Yvonne
  Chinn, Julie~Q. Hang, S{\o}ren Warming, Jackson Egen, Wyne~P. Lee, Cary
  Austin, Yan Wu, Jian Payandeh, John~B. Lowe, and Christian~W. Siebel.
\newblock Therapeutic antibodies reveal notch control of transdifferentiation
  in the adult lung.
\newblock \emph{Nature}, 528:\penalty0 127--131, 2015.
\newblock \doi{10.1038/nature15715}.

\bibitem[Lloyd-Lewis et~al.(2019)Lloyd-Lewis, Mourikis, and
  Fre]{Lloyd-Lewis2019}
Bethan Lloyd-Lewis, Philippos Mourikis, and Silvia Fre.
\newblock Notch signalling: sensor and instructor of the microenvironment to
  coordinate cell fate and organ morphogenesis.
\newblock \emph{Current Opinion in Cell Biology}, 61:\penalty0 16--23, 2019.
\newblock \doi{10.1016/j.ceb.2019.06.003}.

\bibitem[Simian and Bissell(2017)]{Simian2017}
Marina Simian and Mina~J. Bissell.
\newblock Organoids: A historical perspective of thinking in three dimensions.
\newblock \emph{Journal of Cell Biology}, 216:\penalty0 31--40, 2017.
\newblock \doi{10.1083/jcb.201610056}.

\bibitem[Bouras et~al.(2008)Bouras, Pal, Vaillant, Harburg, Asselin-Labat,
  Oakes, Lindeman, and Visvader]{Bouras2008}
Toula Bouras, Bhupinder Pal, Fran{\c{c}}ois Vaillant, Gwyndolen Harburg,
  Marie~Liesse Asselin-Labat, Samantha~R. Oakes, Geoffrey~J. Lindeman, and
  Jane~E. Visvader.
\newblock Notch signaling regulates mammary stem cell function and luminal
  cell-fate commitment.
\newblock \emph{Cell Stem Cell}, 3:\penalty0 429--441, 2008.
\newblock \doi{10.1016/j.stem.2008.08.001}.

\bibitem[Lafkas et~al.(2013)Lafkas, Rodilla, Huyghe, Mourao, Kiaris, and
  Fre]{Lafkas2013}
Daniel Lafkas, Veronica Rodilla, Mathilde Huyghe, Larissa Mourao, Hippokratis
  Kiaris, and Silvia Fre.
\newblock Notch3 marks clonogenic mammary luminal progenitor cells in vivo.
\newblock \emph{Journal of Cell Biology}, 203:\penalty0 47--56, 2013.
\newblock \doi{10.1083/jcb.201307046}.

\bibitem[Toffoli and Margolus(1987)]{toffoli1987cellular}
Tommaso Toffoli and Norman Margolus.
\newblock \emph{Cellular automata machines: a new environment for modeling},
  chapter~5, pages 60--65.
\newblock MIT press, 1987.

\bibitem[Bollob{\'a}s(2013)]{bollobas2013modern}
B{\'e}la Bollob{\'a}s.
\newblock \emph{Modern graph theory}, volume 184, chapter~1.
\newblock Springer Science \& Business Media, 2013.

\bibitem[Godsil(1997)]{Godsil1997}
Chris~D. Godsil.
\newblock Compact graphs and equitable partitions.
\newblock \emph{Linear Algebra and Its Applications}, 255:\penalty0 259--266,
  1997.
\newblock \doi{10.1016/S0024-3795(97)83595-1}.

\bibitem[Weiss(1994)]{Weiss1994}
George Weiss.
\newblock Transfer functions of regular linear systems. part i:
  Characterizations of regularity.
\newblock \emph{Transactions of the American Mathematical Society},
  342:\penalty0 827--854, 1994.
\newblock \doi{10.2307/2154655}.

\bibitem[Angeli and Sontag(2003)]{Angeli2003}
David Angeli and Eduardo~D. Sontag.
\newblock Monotone control systems.
\newblock \emph{IEEE Transactions on Automatic Control}, 48:\penalty0
  1684--1698, 2003.
\newblock \doi{10.1109/TAC.2003.817920}.

\bibitem[van~der Schaft(1996)]{DerSchaft1996}
Arjan~J. van~der Schaft.
\newblock \emph{$\mathcal{L}_{2}$ Gain and Passivity Techniques in Nonlinear
  Control}, chapter~1.
\newblock 1996.

\bibitem[Doyle et~al.(2013)Doyle, Francis, and Tannenbaum]{doyle2013feedback}
John~C. Doyle, Bruce~A. Francis, and Allen~R. Tannenbaum.
\newblock \emph{Feedback control theory}.
\newblock Courier Corporation, 2013.

\bibitem[Angeli and Sontag(2004)]{Angeli2004}
David Angeli and Eduardo~D. Sontag.
\newblock \emph{Interconnections of Monotone Systems with Steady-State
  Characteristics}, pages 135--154.
\newblock 2004.
\newblock \doi{10.1007/978-3-540-39983-4_9}.

\bibitem[Chang et~al.(2008)Chang, Pearson, and Zhang]{Chang2008}
K.~Ching Chang, Kelly Pearson, and Tan Zhang.
\newblock Perron-frobenius theorem for nonnegative tensors.
\newblock \emph{Communications in Mathematical Sciences}, 6:\penalty0 507--520,
  2008.
\newblock \doi{10.4310/CMS.2008.v6.n2.a12}.

\bibitem[Haddad and Chellaboina(2011)]{haddad2011nonlinear}
Wassim~M. Haddad and VijaySekhar Chellaboina.
\newblock \emph{Nonlinear dynamical systems and control: a Lyapunov-based
  approach}.
\newblock Princeton university press, 2011.

\bibitem[Mirams et~al.(2013)Mirams, Arthurs, Bernabeu, Bordas, Cooper, Corrias,
  Davit, Dunn, Fletcher, Harvey, Marsh, Osborne, Pathmanathan, Pitt-Francis,
  Southern, Zemzemi, and Gavaghan]{Mirams2013}
Gary~R. Mirams, Christopher~J. Arthurs, Miguel~O. Bernabeu, Rafel Bordas,
  Jonathan Cooper, Alberto Corrias, Yohan Davit, Sara~Jane Dunn, Alexander~G.
  Fletcher, Daniel~G. Harvey, Megan~E. Marsh, James~M. Osborne, Pras
  Pathmanathan, Joe Pitt-Francis, James Southern, Nejib Zemzemi, and David~J.
  Gavaghan.
\newblock Chaste: An open source c++ library for computational physiology and
  biology.
\newblock \emph{PLoS Computational Biology}, 9:\penalty0 e1002970, 2013.
\newblock \doi{10.1371/journal.pcbi.1002970}.

\bibitem[Osborne et~al.(2017)Osborne, Fletcher, Pitt-Francis, Maini, and
  Gavaghan]{Osborne2017}
James~M. Osborne, Alexander~G. Fletcher, Joe~M. Pitt-Francis, Philip~K. Maini,
  and David~J. Gavaghan.
\newblock Comparing individual-based approaches to modelling the
  self-organization of multicellular tissues.
\newblock \emph{PLOS Computational Biology}, 13:\penalty0 e1005387, 2017.
\newblock \doi{10.1371/journal.pcbi.1005387}.

\bibitem[Atwell et~al.(2015)Atwell, Qin, Gavaghan, Kugler, Hubbard, and
  Osborne]{Atwell2015}
Kathryn Atwell, Zhao Qin, David Gavaghan, Hillel Kugler, E.~Jane~Albert
  Hubbard, and James~M. Osborne.
\newblock Mechano-logical model of c. elegans germ line suggests feedback on
  the cell cycle.
\newblock \emph{Development (Cambridge)}, 142\penalty0 (22):\penalty0
  3902--3911, 2015.
\newblock \doi{10.1242/dev.126359}.

\bibitem[Purcell(1977)]{purcell1977life}
Edward~M. Purcell.
\newblock Life at low reynolds number.
\newblock \emph{American journal of physics}, 45:\penalty0 3--11, 1977.

\bibitem[Chapra(2012)]{chapra2012applied}
Steven~C. Chapra.
\newblock \emph{Applied numerical methods with MATLAB for engineers and
  scientists}, chapter~23, pages 555--572.
\newblock New York: McGraw-Hill, 2012.

\bibitem[Pathmanathan et~al.(2009)Pathmanathan, Cooper, Fletcher, Mirams,
  Murray, Osborne, Pitt-Francis, and Walter]{pathmanathan2009computational}
Pras Pathmanathan, J.~Cooper, Alexander~G. Fletcher, Gary~R. Mirams, Philip~K.
  Murray, James~M. Osborne, Joe Pitt-Francis, and Jon~S. Walter, Alex
  .and~Chapman.
\newblock A computational study of discrete mechanical tissue models.
\newblock \emph{Physical biology}, 6:\penalty0 036001, 2009.
\newblock \doi{10.1088/1478-3975/6/3/036001}.

\bibitem[Gonz{\'a}lez(2010)]{gonzalez2010measurement}
{\'A}lvaro Gonz{\'a}lez.
\newblock Measurement of areas on a sphere using fibonacci and
  latitude-longitude lattices.
\newblock \emph{Mathematical Geosciences}, 42:\penalty0 49, 2010.
\newblock \doi{10.1007/s11004-009-9257}.

\bibitem[Wearing et~al.(2000)Wearing, Owen, and Sherratt]{Wearing2000}
Helen~J. Wearing, Markus~R. Owen, and Jonathan~A. Sherratt.
\newblock Mathematical modelling of juxtacrine patterning.
\newblock \emph{Bulletin of Mathematical Biology}, 62:\penalty0 293--320, 2000.
\newblock \doi{10.1006/bulm.1999.0152}.

\bibitem[Knudsen and Wheelock(2005)]{Knudsen2005}
Karen~A. Knudsen and Margaret~J. Wheelock.
\newblock Cadherins and the mammary gland.
\newblock \emph{Journal of Cellular Biochemistry}, 95:\penalty0 488--496, 2005.
\newblock \doi{10.1002/jcb.20419}.

\bibitem[Pereira et~al.(2006)Pereira, Teixeira, Pinho, Ferreira, Fernandes,
  Oliveira, Seruca, Suriano, and Casares]{Pereira2006}
Paulo~S. Pereira, Alexandra Teixeira, Sofia Pinho, Paulo Ferreira, Joana
  Fernandes, Carla Oliveira, Raquel Seruca, Gianpaolo Suriano, and Fernando
  Casares.
\newblock E-cadherin missense mutations, associated with hereditary diffuse
  gastric cancer (hdgc) syndrome, display distinct invasive behaviors and
  genetic interactions with the wnt and notch pathways in drosophila epithelia.
\newblock \emph{Human Molecular Genetics}, 15:\penalty0 1704--1712, 2006.
\newblock \doi{10.1093/hmg/ddl093}.

\bibitem[Leong et~al.(2007)Leong, Niessen, Kulic, Raouf, Eaves, Pollet, and
  Karsan]{Leong2007}
Kevin~G. Leong, Kyle Niessen, Iva Kulic, Afshin Raouf, Connie Eaves, Ingrid
  Pollet, and Aly Karsan.
\newblock Jagged1-mediated notch activation induces epithelial-to-mesenchymal
  transition through slug-induced repression of e-cadherin.
\newblock \emph{Journal of Experimental Medicine}, 204:\penalty0 2935--2948,
  2007.
\newblock \doi{10.1084/jem.20071082}.

\bibitem[Chen et~al.(2010)Chen, Imanaka, and Griffin]{chen2010hypoxia}
Jing Chen, Naoko Imanaka, and James~D. Griffin.
\newblock Hypoxia potentiates notch signaling in breast cancer leading to
  decreased e-cadherin expression and increased cell migration and invasion.
\newblock \emph{British journal of cancer}, 102:\penalty0 351--360, 2010.

\bibitem[Shamir et~al.(2014)Shamir, Pappalardo, Jorgens, Coutinho, Tsai, Aziz,
  Auer, Tran, Bader, and Ewald]{Shamir2014}
Eliah~R. Shamir, Elisa Pappalardo, Danielle~M. Jorgens, Kester Coutinho,
  Wen~Ting Tsai, Khaled Aziz, Manfred Auer, Phuoc~T. Tran, Joel~S. Bader, and
  Andrew~J. Ewald.
\newblock Twist1-induced dissemination preserves epithelial identity and
  requires e-cadherin.
\newblock \emph{Journal of Cell Biology}, 204:\penalty0 839--856, 2014.
\newblock \doi{10.1083/jcb.201306088}.

\end{thebibliography}

\end{document}